\documentclass[a4paper,11pt]{article}
\pdfoutput=1
\usepackage[utf8]{inputenc}
\usepackage[T1]{fontenc}
\usepackage[english]{babel}
\usepackage{csquotes}
\usepackage[margin=2.5cm]{geometry}

\usepackage{amsmath,amsthm}
\usepackage{amssymb}
\usepackage{mathtools}
\usepackage{xfrac}
\allowdisplaybreaks[2]

\usepackage{newtxtext,newtxmath}
\usepackage{microtype}
\usepackage{setspace}
\usepackage{graphicx}
\usepackage{tikz}
\usetikzlibrary{arrows.meta}
\usetikzlibrary{positioning}
\graphicspath{{./figures/}}
\usepackage{float}
\usepackage[section]{placeins} %
\usepackage[labelfont=bf,labelsep=period]{caption}
\usepackage{subcaption}

\usepackage{array}
\usepackage{booktabs}

\usepackage{xcolor}
\usepackage{tcolorbox}
\tcbuselibrary{theorems,skins,breakable}

\tcbset{
  protocol/.style={
    enhanced,
    boxrule=0.6pt,
    colback=gray!5,
    colframe=black,
coltitle=black,
    title={#1},
    sharp corners,
    before upper={\vspace{2mm}\parindent15pt}, after title={\par\smallskip},
    attach boxed title to top left={
      yshift=-3mm, xshift=4mm
    },
    boxed title style={
      colback=gray!5,
      colframe=black,
      size=small,
      boxrule=0.6pt,
    }
  }
}

\usepackage[shortlabels]{enumitem}
\usepackage{fmtcount}
\setlist[itemize]{itemsep=0pt,parsep=0pt,topsep=0pt,partopsep=0pt}

\newcommand{\step}[1]{\noindent\underline{\textsc{#1}}\par}
\newcommand{\stepmin}[1]{\vspace{-0.5\baselineskip}\noindent\underline{#1}\par}

\newlist{codelines}{enumerate}{1}
\setlist[codelines,1]{
  label=\texttt{\padzeroes[2]{\arabic*}:}, leftmargin=2.5em,
  labelwidth=0em,
  labelsep=1.0em,
  itemsep=0pt,
  parsep=0em,
  font=\ttfamily
}

\usepackage{url}
\usepackage[colorlinks=true,linkcolor=black,citecolor=blue,urlcolor=blue]{hyperref}
\usepackage{bookmark}
\usepackage{orcidlink}
\usepackage{fontawesome5}

\usepackage[style=alphabetic,backend=biber,maxbibnames=99]{biblatex}
\usepackage{xspace}
\usepackage[nodayofweek]{datetime}
\usepackage{authblk}

\usepackage{todonotes}

\floatstyle{ruled}
\newfloat{algo}{htbp}{algo}
\floatname{algo}{Algorithm}
\newfloat{interface}{htbp}{interface}
\floatname{interface}{Interface}
\theoremstyle{definition}
\newtheorem{definition}{Definition}[section]
\newtheorem{example}[definition]{Example}

\theoremstyle{plain}
\newtheorem{theorem}[definition]{Theorem}
\newtheorem{lemma}[definition]{Lemma}

\newtheorem{corollary}[definition]{Corollary}

\usepackage{etoolbox}
\AtBeginEnvironment{definition}{\begin{samepage}}
\AtEndEnvironment{definition}{\end{samepage}}

\newcommand{\etal}{\emph{et al.}\xspace}
\renewcommand{\paragraph}[1]{\par\medskip\noindent\textbf{#1}\enspace}

\newcommand{\set}[1]{\ensuremath{\{#1\}}}

\newcommand{\true}{\ensuremath{\texttt{TRUE}}\xspace}
\newcommand{\false}{\ensuremath{\texttt{FALSE}}\xspace}
\newcommand{\consnum}[1]{\ensuremath{\mathsf{cons}(#1)}\xspace}
\newcommand{\op}[1]{\ensuremath{\textsl{#1}}\xspace}
\newcommand{\tr}[1]{\ensuremath{\op{transfer}(#1)}\xspace}
\newcommand{\rd}[1]{\ensuremath{\op{read}(#1)}\xspace}

\newcommand{\att}{\ensuremath{\mathsf{\mathcal{E}}}\xspace}
\newcommand{\cuat}{\ensuremath{\mathsf{CUAT}}\xspace}
\newcommand{\luat}{\ensuremath{\mathsf{LUAT}}\xspace}

\setenumerate[1]{label={(\arabic*)}}
\setitemize[1]{label=\textbullet}

\usepackage{silence}
\title{\textbf{The Consensus Number of Untraceable Cryptocurrencies}}
\author[1]{Christian Cachin\,\orcidlink{0000-0001-8967-9213}}
\author[1]{David Lehnherr\,\orcidlink{0000-0002-4956-4064}}
\author[1]{Juan Villacis\,\orcidlink{0009-0006-0110-8613}}
\author[1]{François-Xavier Wicht\,
  \href{mailto:francois-xavier.wicht@unibe.ch}{\textcolor{black}{\footnotesize\faEnvelope}}\,
  \orcidlink{0009-0005-6090-7901}}
\affil[1]{Institute of Computer Science, University of Bern, Switzerland}
\date{}

\begin{document}
\maketitle

\begin{abstract}\noindent
Sender untraceability hides the account spent by a cryptocurrency transfer
among a set of candidates, its \emph{masking set}.  What a transfer does to
that set separates two designs: classical schemes retain the whole set and
append a nullifier marking the spent account, so the ledger grows with every
transfer; constant-state schemes instead consume and replace the entire set.
We ask how this choice affects synchronization.

We formalize the two designs as the \emph{linear untraceable asset transfer}
(LUAT) and \emph{constant untraceable asset transfer} (CUAT) objects and locate
them in the consensus hierarchy.  In LUAT, transfers from distinct accounts
commute.  Its consensus number is $2$, compared with $1$ for standard asset
transfer, independently of the masking-set size and of the untraceability
notion, and LUAT is starvation-free.  Partitioning the accounts into fixed
masking sets lets exhausted sets be garbage-collected without increasing that
number.

In CUAT, a transfer consumes and replaces every account of its masking set, so
two transfers whose sets intersect cannot both take effect.  We formalize this
with the \emph{conflict graph} on masking sets, whose edges join sets sharing an
account.  Under weak untraceability, which protects a transaction in isolation,
the consensus number is unbounded already for one-round protocols.  Under strong
untraceability, which protects against an observer of the complete history,
untraceability holds on a history exactly when any two accounts sharing a
masking set occur in the same number of the masking sets in it.  This uniform
incidence bounds the conflict graph, and matching constructions attain it, so
the consensus number is determined exactly and grows quadratically in the
masking-set size.  Finally, CUAT is not starvation-free.  The two objects
therefore pay for the same privacy differently: LUAT in storage, CUAT in
synchronization and fairness.
\end{abstract}

\newpage
\tableofcontents
\newpage

\section{Introduction}
\label{sec:intro}

Cryptocurrencies enable peer-to-peer value transfer without trusted
intermediaries. At their core, these systems maintain a \emph{ledger}, a
replicated data structure recording the transaction history and the current
state of all accounts. While systems such as Bitcoin provide
pseudonymity, privacy-preserving cryptocurrencies aim for stronger guarantees.
In particular, \emph{sender untraceability} prevents an external observer from
identifying which account initiated a transaction within a set of candidate
accounts, called here the \emph{masking set}. Achieving this property while
preventing double spending poses a
fundamental challenge for system design.

The standard approach to sender untraceability relies on two complementary data
structures: an \emph{allow-set} containing all potentially unspent accounts and a
\emph{deny-set} of nullifiers marking spent accounts. To transfer value, a user
proves in zero knowledge that an account in the allow-set corresponds to a
previously unused nullifier, without revealing which one. The limitation of this
approach is indefinite state growth: removing an account may reveal that it was
spent, whereas removing its nullifier may permit a double spend. Partitioning
the accounts into pre-declared masking sets of bounded size recovers part of the
loss. A set of $k$ accounts admits at most $k$ transfers, after which it is
exhausted and both the set and its nullifiers can be
collected~\cite{DBLP:conf/csfw/ChowELRW23,DBLP:journals/popets/CachinW25}. This
slows the growth without bounding it, and it adds a second shared object to the
system, whose synchronization cost we also determine.

Recent work avoids this growth through account re-randomization. Quisquis and
related systems~\cite{DBLP:conf/asiacrypt/FauziMMO19,DBLP:journals/iacr/AlupothaGC24,DBLP:conf/asiacrypt/MadathilS25}
invalidate and replace every account in a fixed-size masking set whenever one of
them transfers value. Cryptographic randomization hides which balances changed,
and the ledger retains only the current account representations. The resulting
state is bounded by the number of active accounts, but two transfers whose
masking sets overlap cannot both succeed. Three questions follow, and the paper
answers all three.
\begin{enumerate}
  \item \emph{What synchronization power does each design provide?}
  \item \emph{What does untraceability itself cost?}
  \item \emph{What is traded when unbounded state is replaced by masking-set
    contention?}
\end{enumerate}

The synchronization consequences of this distinction are not captured by the
cryptographic security definitions alone.  We study them through the consensus
hierarchy~\cite{DBLP:journals/toplas/Herlihy91}, building on the shared-memory
formulation of asset transfer by Guerraoui
\etal~\cite{DBLP:journals/dc/GuerraouiKMPS22}.  The consensus number of a shared
object is the largest number of processes that can solve wait-free consensus
using instances of that object and atomic registers.  Standard asset transfer
has consensus number one: transfers on independently owned accounts require no
agreement on their relative order.  Consensus numbers therefore provide a
natural language for asking whether sender untraceability preserves this
independence, and how the answer changes when the anonymity mechanism is also
used to keep the state bounded.

This perspective also captures the parallelism available to a cryptocurrency
implementation. Systems such as Sui and Zef exploit the consensus number one of
standard asset transfer~\cite{DBLP:journals/dc/GuerraouiKMPS22} by processing
independent accounts without a global total order and invoking stronger
coordination only when accounts are shared.
For private transfers, a higher consensus number identifies a larger potential
contention structure and therefore limits the same form of parallel commit and
shardable verification. Our bounds measure how much of this independence is
preserved by the two untraceability mechanisms.

We formalize the two designs as shared objects parameterized by the masking-set
size $\lambda$.  The \emph{linear untraceable asset transfer} object
$\luat_\lambda$ models the allow-set/deny-set approach.  A transfer consumes one
sender account, creates previously unused account representations, and names the remaining members of its
masking set only as decoys.  The \emph{constant untraceable asset transfer}
object $\cuat_\lambda$\footnote{Pronounced \textit{Coo-at}.} models account
re-randomization.  A transfer atomically consumes and replaces every account in
its masking set, although only the sender and recipient balances change.  In
both objects the sender is hidden among $\lambda$ accounts; the distinction is
that masking-set overlap is inert in $\luat_\lambda$ and constitutes a conflict
in $\cuat_\lambda$.

This distinction already determines the behavior of the linear-state object.
We prove that a $\luat_\lambda$ object has consensus number exactly
$2$, independently of $\lambda$.  Transfers spending distinct accounts commute
and both succeed, so increasing the masking-set size creates no additional
contention.  The object nevertheless exceeds the consensus number of standard
asset transfer because a successful operation observes the resulting deny-set
and thereby learns its position among concurrent transfers.  When $k$ processes
share an account, contention over its unique nullifier raises the consensus
number to $\max(k,2)$.  

The same separation persists in the progress and storage properties of
$\luat_\lambda$.  A transfer that is eligible when invoked cannot be invalidated
by activity on accounts owned by other processes, and the object is therefore
starvation-free.  Its cost is the monotone growth of the allow-set and deny-set.
We formalize the partitioning object used to garbage-collect exhausted,
pre-declared masking sets and determine its consensus number.  Along the
executions relevant to garbage collection, the object is implementable from
fetch-and-add and implements two-process consensus; its consensus number is
therefore exactly $2$ for every part size at least two.  Partitioning does not
increase the synchronization power required by the linear-state design.

The analysis of $\cuat_\lambda$ is governed instead by the intersection pattern
of its masking sets.  We associate with a protocol a conflict graph whose
vertices are the masking sets it may use and whose edges join sets that share an
account index.  Adjacent transfers cannot both succeed, whereas non-adjacent
transfers commute.  For \emph{one-round} protocols, in which every process
invokes one transfer, this graph gives an exact characterization: a clique of
size $d$ implements $d$-process consensus by an attempt-and-adopt protocol, and
the masking sets of a one-round consensus protocol form a clique at a critical
configuration.  Hence the one-round consensus power equals the largest
admissible clique.  This regime captures the synchronization within a single
payment attempt, where the masking set is fixed before submission.

Unrestricted protocols can use $\cuat_\lambda$ in several rounds and choose
later masking sets as a function of earlier outcomes.  The relevant structure
at a critical configuration is then weaker than a clique: pending transfers
from opposite valency classes must intersect, producing a complete bipartite
subgraph of the conflict graph, but transfers within one class need not
intersect.  This distinction is essential.  It separates the synchronization
required by a single transaction from the full consensus power of the object,
and it allows multi-round protocols to solve consensus among strictly more
processes than any one-round protocol.

We consider two notions of sender untraceability, distinguished by the
adversary's view.  Weak untraceability protects each transaction in isolation.
We write $\cuat_\lambda^-$ for CUAT under this requirement and
$\cuat_\lambda^+$ for CUAT under strong untraceability.
Weak untraceability imposes no global restriction on masking-set incidence, and a single hub
account can occur in arbitrarily many masking sets.  These sets form
arbitrarily large cliques, so $\consnum{\cuat_\lambda^-}=\infty$ already for
one-round protocols.  Strong untraceability protects against an adversary that
observes the complete transaction history and uses repeated account appearances
to update its beliefs.  We prove that strong untraceability holds on a history
exactly when any two accounts sharing a masking set appear in the same number of
them; on the pairwise-intersecting families a consensus protocol uses, this
makes the incidence equal across every account occurring in the history.  This uniform-incidence characterization turns the
privacy requirement into a combinatorial constraint on the conflict graph.

For one-round protocols, uniform incidence together with pairwise intersection
bounds the number of masking sets by $\lambda^2-\lambda+1$.  When
$\lambda-1$ is a prime power, the lines of $\mathrm{PG}(2,\lambda-1)$, the
projective plane of order $\lambda-1$, meet the bound: every line contains $\lambda$ points,
every two lines intersect, and every point has the same incidence.  For general
$\lambda\ge3$, we construct a cyclic difference cover with
$\lfloor\lambda^2/2\rfloor$ translates.  It yields a uniformly incident,
pairwise-intersecting family and therefore a one-round protocol of that size.
Thus the one-round consensus number under strong untraceability lies between
$\lfloor\lambda^2/2\rfloor$ and $\lambda^2-\lambda+1$, with equality at the
upper bound whenever $\lambda-1$ is a prime power.

For unrestricted protocols, we determine the consensus number exactly:
$\consnum{\cuat_\lambda^+}=4\lfloor\lambda/2\rfloor\lceil\lambda/2\rceil$ for
every $\lambda\ge3$, that is, $\lambda^2$ when $\lambda$ is even and
$\lambda^2-1$ when it is odd.  The upper bound follows from a refined incidence
count over the complete bipartite intersection at a critical configuration.
For the matching lower bound, two groups of processes first run the cyclic
one-round protocol independently and then reconcile their group decisions on a
common object.  The second round uses a cross-cover family in which every set
from one group intersects every set from the other while all accounts retain
uniform incidence.  A grid supplies such a family for even $\lambda$; for odd
$\lambda$ we give a cyclic construction that alternates the two possible
balanced incidence patterns.  The resulting protocol attains the upper bound
without any arithmetic condition on $\lambda$.  In particular, the
multi-round consensus number strictly exceeds the one-round upper bound for
every $\lambda\ge3$.

Finally, the conflict structure that gives $\cuat_\lambda$ its synchronization
power also prevents starvation-freedom under asynchronous scheduling.  A
process can repeatedly issue zero-value transfers that re-randomize an
overlapping masking set before a competitor succeeds, while remaining eligible
indefinitely.  This establishes the denial-of-service behavior conjectured for
Quisquis~\cite{DBLP:conf/fc/BunzAZB20}.  The contrast with $\luat_\lambda$ is
again structural: a linear-state transfer cannot invalidate a masking set that
it merely names, whereas a constant-state transfer invalidates every account in
the set it uses.

The two objects therefore satisfy the same uniform-incidence condition for
strong untraceability but pay for it differently.  In $\luat_\lambda$,
uniformity constrains the masking sets without affecting synchronization.  In
$\cuat_\lambda$, the masking sets simultaneously provide anonymity, bound the
state, and determine which transfers conflict; the exact consensus number is
quadratic in $\lambda$.  The additional synchronization cost is thus not a
consequence of sender privacy alone.  It arises when the mechanism that hides
the sender is also made responsible for reclaiming state.

From an operational perspective, the distinction is local rather than global.
Disjoint CUAT transfers still commute and may be processed independently; the
additional coordination appears only among transfers connected in the conflict
graph. The parameter $\lambda$ nevertheless controls both the size of the
anonymity set and the worst-case synchronization power of that contention
structure. LUAT separates these two roles at the cost of retaining historical
state, whereas CUAT couples them in order to reclaim it. The consensus hierarchy
therefore exposes a design choice that is hidden at the cryptographic layer:
private payment systems may pay for sender untraceability in storage, or use the
same masking sets to bound storage and pay instead in synchronization and
progress guarantees. We treat the two cases in parallel in
Sections~\ref{sec:luat} and~\ref{sec:ruat}, over the untraceability notions that
Section~\ref{sec:untraceability} defines and characterizes by uniform
masking-set incidence, and Section~\ref{sec:parallel-fairness} turns the same
conflict structure on parallelization and fairness.

\section{Related work}\label{sec:rel-work}
Our analysis of CUAT's synchronization power connects four research areas:
consensus numbers of cryptocurrency objects and privacy-preserving primitives,
quorum systems and combinatorial design theory, multi-word compare-and-swap
primitives, and the storage cost of authenticated data structures.
Table~\ref{tab:results} collects the consensus numbers of the objects discussed
below, together with the ones this paper establishes.

\begin{table}[!ht]
  \centering
  \small
  \begin{tabular}{@{}lllcl@{}}
    \toprule
    Object & System & Untraceability & $\mathsf{cons}$ & Source \\
    \midrule
    Asset transfer & Bitcoin, Ethereum & none & $1$ & \cite{DBLP:journals/dc/GuerraouiKMPS22} \\
    $k$-shared asset transfer & Sui, Zef & none & $k$ & \cite{DBLP:journals/dc/GuerraouiKMPS22} \\
    ERC20 token & Tornado Cash, Railgun & whole pool & $m$ & \cite{DBLP:conf/icdcs/AlposCMZ21} \\
    Allow-list & identity, e-voting & --- & $1$ & \cite{DBLP:conf/wdag/FreyGR23} \\
    Deny-list & identity, e-voting & --- & $d$ & \cite{DBLP:conf/wdag/FreyGR23} \\
    \midrule
    $\luat_\lambda$ & Zcash, Monero & weak or strong & $2$ & Theorem~\ref{thm:uat-consensus-number} \\
    $\luat_\lambda$, $k$-shared account & shared account & weak or strong & $\max(k,2)$ & Theorem~\ref{thm:k-uat-consensus-number} \\
    $k$-partitioning & garbage collection & --- & $2$ & Theorem~\ref{thm:part-consensus-number} \\
    \midrule
    $\cuat_\lambda^-$ & \emph{Quisquis} and successors & weak & $\infty$ & Theorem~\ref{thm:weak-untraceability} \\
    $\cuat_\lambda^+$ & \emph{Quisquis} and successors & strong & $4\lfloor\lambda/2\rfloor\lceil\lambda/2\rceil$ & Theorem~\ref{thm:multi-cuat-strong-tight} \\
    \bottomrule
  \end{tabular}
  \caption{Consensus numbers of asset-transfer objects. The parameters are
  distinct: $k$ processes share an account, $m$ spenders are approved on an
  ERC20 token, $d$ processes access the deny-list, $\lambda$ is the masking-set
  size, and the $k$ of $k$-partitioning is a part size, for which the value is
  $1$ when it equals one. The value for $\cuat_\lambda^+$ holds for every
  $\lambda\ge3$ and needs multi-round protocols; one-round protocols reach only
  between $\lfloor\lambda^2/2\rfloor$ and $\lambda^2-\lambda+1$
  (Corollary~\ref{cor:one-round-sandwich}). Only $\cuat_\lambda$ is not
  starvation-free (Theorem~\ref{thm:no-fairness}).}
  \label{tab:results}
\end{table}

Guerraoui \etal~\cite{DBLP:journals/dc/GuerraouiKMPS22} challenge the
assumption that consensus is necessary for cryptocurrencies. By modeling
asset transfers as concurrent objects, they show the consensus
number of standard asset transfer is $1$, which means that transactions can be
processed asynchronously without total ordering. Building on this insight,
modern cryptocurrencies like Sui, Libra/Diem,
and Zef adopt their $k$-shared asset transfer model, where only the
$k$ users sharing an account need consensus, thereby reducing synchronization overhead.
As Alpos \etal~\cite{DBLP:conf/icdcs/AlposCMZ21} observe, ERC20 tokens let an
account owner approve several spenders, which makes the consensus number depend
on how many are authorized. Consensus requirements in smart
contracts are therefore state-dependent rather than fixed.

Privacy-preserving systems, however, face distinct synchronization challenges.
Frey \etal~\cite{DBLP:conf/wdag/FreyGR23} analyze allow-list and
deny-list objects, showing that while allow-lists need no
synchronization, deny-lists require consensus among verifiers
performing set-non-membership proofs.
More broadly, their
framework applies to decentralized identity management and e-voting.

CUAT's pairwise-intersecting masking sets bear structural similarities to
quorum systems~\cite{DBLP:journals/tods/Thomas79,DBLP:conf/sosp/Gifford79}.
Both use intersections among sets of elements to provide synchronization
in distributed systems. The key distinction is that quorum systems
are designed to tolerate process failures and must satisfy constraints
derived from failure assumptions, whereas masking sets in CUAT need not
account for such failures. We show that CUAT requires two conditions beyond
pairwise intersection. First, to implement consensus, the intersections must be
structured so that a losing transfer can identify the winner. Second, to
maintain sender untraceability across a history, incidence must be uniform: any
two accounts sharing a masking set appear in the same number of them.

A CUAT transfer also acts like a multi-word compare-and-swap (MCAS), the
primitive that Harris \etal~\cite{DBLP:conf/wdag/HarrisFP02} implement from
single-word CAS. Both update several locations atomically, here the account
representations of a masking set, and commit only if every location matches its
expected value. The difference is again privacy: an MCAS reveals the locations
it touches, whereas a CUAT transfer must keep its sender untraceable.

\begin{table}[!b]
  \centering
  \small
  \begin{tabular}{@{}llll@{}}
    \toprule
    System & Anonymity mechanism & Ledger state & Object \\
    \midrule
    Zerocoin~\cite{DBLP:conf/sp/MiersG0R13}    & commitment set, serial numbers   & linear   & $\luat_\lambda$ \\
    Zcash~\cite{DBLP:conf/sp/Ben-SassonCG0MTV14} & note commitments, nullifiers   & linear   & $\luat_\lambda$ \\
    Monero~\cite{zerotomonero2020}             & ring signatures, key images      & linear   & $\luat_\lambda$ \\
    Anon.\ Zether~\cite{DBLP:conf/fc/BunzAZB20} & anonymity set, per-epoch nonces & bounded & $\luat_\lambda$ \\
    \midrule
    \emph{Quisquis}~\cite{DBLP:conf/asiacrypt/FauziMMO19} & masking set, rerandomization & constant & $\cuat_\lambda$ \\
    \cite{DBLP:journals/iacr/AlupothaGC24,DBLP:conf/asiacrypt/MadathilS25} & masking set, rerandomization & constant & $\cuat_\lambda$ \\
    \midrule
    Mixers, mixnets, tumblers & off-ledger shuffling or credentials & --- & out of scope \\
    \bottomrule
  \end{tabular}
  \caption{Untraceable payment designs by the effect of a transfer on its
  masking set. The $\luat_\lambda$ designs reference their masking set and
  consume one account; the $\cuat_\lambda$ designs consume the whole set.
  Anonymous Zether discards its deny-set at each epoch boundary, so its state is
  bounded by one epoch of activity rather than fixed, and it is the only listed
  design whose state does not grow with the transaction count without replacing
  accounts.}
  \label{tab:taxonomy}
\end{table}

Christ and Bonneau~\cite{DBLP:conf/fc/ChristB23} introduce revocable proof
systems (RPS) as a unifying abstraction for authenticated data structures,
including stateless blockchains and set commitments. They prove a fundamental
trade-off: any scheme with sublinear global state must incur near-linear
proof-update costs when statements are revoked (e.g., when coins are spent).
This information-theoretic lower bound focuses on storage and witness
maintenance rather than concurrency, but it delimits the same design space we
classify.

Within that space, what a transfer does to its masking set separates
our two objects and sorts the deployed systems. A transfer that consumes one
account and only names the rest of its masking set as decoys leaves a growing
allow-set and deny-set behind, the design of Zerocoin, Zcash and Monero, which
we model as $\luat_\lambda$; a transfer that consumes and replaces its entire
masking set keeps the state bounded, the design of \emph{Quisquis} and its
successors, which we model as $\cuat_\lambda$. Table~\ref{tab:taxonomy} places
the main systems on this division. Anonymity obtained outside the ledger falls
under neither object: mixing services and mixnets relocate the unlinkability to
a shuffling party, and in the credential-based variants the deny-set becomes a
spent-credential set held by a coordinator, so the object carrying the
synchronization is the coordinator rather than the payment system.

\section{Preliminaries}\label{sec:preliminaries}

This section establishes the theoretical foundations for our
analysis. We present the system model, consensus numbers, and the consensus
hierarchy theorem.

\paragraph{Notation.}
We use $\gets$ for assignment in algorithms, $=$ for equality
testing, $\coloneqq$ for definitions, and standard set notation
($\in$ for membership, $\subseteq$ for subset, $\cap$ for
intersection, $\cup$ for union). We write $\binom{S}{k} \coloneqq
\{\,A\subseteq S \mid |A|=k\,\}$ for the family of $k$-element subsets of $S$.

\paragraph{System model.}
We consider $N$ asynchronous sequential processes $p_1, \dots, p_N$
in a set $\mathcal{P}$ that may crash (crash-failure model). Throughout, $\Pi$ denotes a protocol; the process set is $\mathcal{P}$. Processes interact
through shared objects supporting operations that produce responses.
Objects satisfy
\textit{linearizability}~\cite{DBLP:conf/popl/HerlihyW87}, where
operation histories appear as if executed in a single, consistent
order respecting real-time constraints. A \emph{register} is a
shared-memory abstraction supporting $\op{read}$ and $\op{write}$
operations. We assume \emph{atomic registers} with linearizable
semantics, where each operation appears to take effect
instantaneously between invocation and response.

\paragraph{Wait-freedom.}
An implementation of a shared object is \emph{wait-free} if every
operation invoked by a correct process completes in a finite number
of steps, regardless of the behavior of other processes.

\paragraph{Consensus numbers and the consensus hierarchy.}
The consensus number $\consnum{O}$ of an object $O$ is the largest
number $n$ such that it is possible to wait-free implement a
consensus object from atomic registers and objects of type $O$, in a
system of $n$ processes~\cite{DBLP:journals/toplas/Herlihy91}. The
consensus hierarchy shows that objects with higher consensus numbers
are strictly more powerful and cannot be implemented from objects
with lower consensus numbers.

\paragraph{Valency.}
Following the standard
terminology~\cite{DBLP:journals/toplas/Herlihy91}, a configuration
$C$ of a binary consensus protocol is \emph{$v$-valent} (with $v\in\{0,1\}$)
if every execution extending $C$ decides~$v$, \emph{univalent} if it
is $v$-valent for some $v$, and \emph{bivalent} if both decision
values are still reachable from $C$. We write $C\cdot p$ for the
configuration reached when process $p$ takes one step from $C$, and
chain steps as $C\cdot p\cdot q$ for the configuration reached when
$p$ then $q$ each take one step. A bivalent configuration is
\emph{critical} if every immediate successor $C\cdot p$ is univalent.

\paragraph{Ownership.}
Both objects we study attach an owner to each account through a map
$\mu:\mathcal{A}\to 2^{\mathcal{P}}\setminus\{\emptyset\}$, and only a
process in $\mu(a)$ may spend $a$ or read its balance. An account $a$ is
\emph{$k$-shared} if $|\mu(a)|=k$, and \emph{unshared} if $|\mu(a)|=1$;
an object is unshared if all of its accounts are. Following Guerraoui
\etal~\cite{DBLP:journals/dc/GuerraouiKMPS22}, sharing is the mechanism by
which several processes contend for one account.

\paragraph{Account-representation uniqueness.}
Every transfer of either object creates account representations. We require
each newly created representation to differ from every representation already
used by the system, and assume that distinct processes never generate the same
representation. In a concrete instantiation, account representations embed
independent randomness, so a collision occurs only with negligible probability
in the security parameter. This condition is required for security, and in
particular for integrity: reusing a representation could conflate distinct
ownership and balance histories or make an old account representation valid
again. It also keeps the concurrency model
faithful to the payment abstraction, since otherwise two processes could
contend by proposing the same representation rather than by operating on the
same account. In particular, the account sets created by concurrent transfers
are disjoint.

\section{Sender untraceability}
\label{sec:untraceability}
We consider an external adversary \att that eavesdrops on network ciphertexts,
records shared-object states, and analyzes process interactions to infer
relationships among accounts. We model \att as computationally unbounded while
idealizing the cryptographic primitives: in particular, the zero-knowledge
proof reveals no information about its witness. Thus, for each transfer, \att
observes the masking set but not the sender account hidden by the witness.

The definitions below apply to any \emph{untraceable asset-transfer object}: a
shared object whose transfers each name a masking set of accounts, one of which
funds the transfer.

\begin{definition}[Masking sets and $\lambda$-masking]
  \label{def:masking-set}
  A \emph{masking set} is a set $S\subseteq\mathcal{A}$ named by a transfer,
  exactly one of whose members funds it; the remaining members are
  \emph{decoys}. An object is \emph{$\lambda$-masking} if every masking set it
  uses has size exactly $\lambda$.
\end{definition}

We assume throughout that masking sets have at least $\lambda\ge2$ members.
Nothing else about the object matters here; in particular it is irrelevant
whether a transfer leaves its masking set intact or replaces it, and
Sections~\ref{sec:luat} and~\ref{sec:ruat} introduce one object of each kind.

Untraceability prevents \att from identifying the funding account. We
distinguish two views: a single transfer and the complete transfer history,
including failed transfers. We write $P(a_i\mid\cdot)$ for \att's posterior
probability that account $a_i$ is the sender account of the transfer under
consideration.

\begin{definition}[Weak untraceability]
  \label{def:weak-untraceability}
  \label{def:luat-weak-untraceability}
  An untraceable asset-transfer object satisfies \emph{weak untraceability} if,
  for every transfer $\tau$ with masking set $S$ and every $a_i\in S$,
  \[
    P(a_i\mid\tau)=1/|S|.
  \]
\end{definition}

The first notion fixes \att's view to one transfer in isolation. Widening it to
everything \att has recorded, including the masking sets of transfers that
failed, gives the second.

\begin{definition}[Strong untraceability]
  \label{def:strong-untraceability}
  \label{def:luat-strong-untraceability}
  An untraceable asset-transfer object satisfies \emph{strong untraceability}
  if, for every transfer $\tau$ with masking set $S$ in every history $H$ and
  every $a_i\in S$,
  \[
    P(a_i\mid H)=1/|S|.
  \]
\end{definition}

Weak untraceability is the standard guarantee when a transaction is considered
in isolation~\cite{DBLP:conf/sp/MiersG0R13,DBLP:conf/sp/Ben-SassonCG0MTV14};
strong untraceability additionally protects against statistical inference from
repeated account appearances~\cite{DBLP:journals/popets/MoserSHLHSHHMNC18,DBLP:journals/popets/EggerLRWY22}.
The second notion implies the first: averaging $P(a_i\mid H)=1/|S|$ over the
histories consistent with an observed $\tau$ gives $P(a_i\mid\tau)=1/|S|$.
Since every masking set has at least $\lambda$ members, both notions bound the
adversary's success probability by $1/\lambda$. For an object $O$ we write $O^-$
for its weak variant and $O^+$ for its strong one; the notion in force restricts
the admissible executions of the object, and no other component of its
specification changes.

Strong untraceability admits a combinatorial characterization that depends on
the masking sets alone. For a history $H$, let $f_a(H)$ denote the number of masking sets in
which account $a$ appears, and let
$\mathcal{A}(H)=\{a:f_a(H)\ge1\}$ denote the accounts occurring in the history.
We use a symmetric worst-case adversary model. Each account has an unknown
probability of funding a transfer and, in objects whose masking sets also
contain the recipient, an unknown probability of receiving one. Decoys are
sampled uniformly. Before observing $H$, the
adversary treats all accounts alike; after observing it, the adversary updates
its beliefs using the account-incidence counts. Transfers are issued
independently. Side channels such as network metadata and timing are outside the
model.

\begin{theorem}[Uniformity]
  \label{thm:untraceability-uniform}
  \label{thm:luat-untraceability-uniform}
  Under the adversary model above, an untraceable asset-transfer object satisfies
  strong untraceability on a history $H$ if and only if $f_i(H)=f_j(H)$ for every
  two accounts $a_i,a_j$ lying in a common masking set of $H$.
\end{theorem}

\begin{proof}
Write $A=\mathcal{A}(H)$. An account enters the masking set $S_\tau$ of a
transfer $\tau$ in exactly one of three ways: it funds $\tau$, which has unknown
probability $q_a$; it receives $\tau$, of unknown probability $r_a$, on objects
that place the recipient in the masking set; or it is one of the decoys, which
the sender draws uniformly, so that this last event has the same probability for
every account alike. Transfers are issued independently, and an account's
identity bears on these events only through $(q_a,r_a)$, so the probability of
$H$ depends on an account only through the number $f_a(H)$ of masking sets
containing it:
\[
  \Pr\big[H\mid (q_a,r_a)_{a\in A}\big]=\Pr\big[H'\mid (q_a,r_a)_{a\in A}\big]
  \qquad\text{whenever } f_a(H)=f_a(H')\ \text{for every } a\in A.
\]
Thus $f_a(H)$ is a sufficient statistic for $(q_a,r_a)$. The prior is symmetric,
so relabelling the accounts by any bijection that leaves every count unchanged
leaves the adversary's beliefs unchanged as well; two accounts with the same
count are therefore interchangeable to it. Two accounts with different counts are
not: decoy draws contribute equally to every count, so of two accounts the one
occurring more often is the one the adversary rates more likely to have funded a
transfer, and $\mathbb{E}[q_a\mid H]$ grows strictly with $f_a(H)$.

\emph{Uniform incidence implies strong untraceability.} Suppose
$f_a(H)=f_b(H)$ for all $a,b\in A$. Fix $\tau\in H$ and $a_i,a_j\in S_\tau$, and
note $S_\tau\subseteq A$. Let $\pi$ swap $a_i$ and $a_j$ and fix every other
account. It maps $S_\tau$ to itself and $H$ to a history carrying the same
counts, which the adversary therefore judges as it judges $H$, while taking the
event ``$a_i$ funded $\tau$'' to ``$a_j$ funded $\tau$''. The two events have
equal posterior probability. Exactly one member of $S_\tau$ funded $\tau$, so
these $|S_\tau|$ probabilities sum to $1$ and each equals $1/|S_\tau|$. Strong
untraceability holds on $H$.

\emph{Strong untraceability implies uniform incidence.} Suppose strong
untraceability holds on $H$, and fix $\tau\in H$ with $a_i,a_j\in S_\tau$. Both
have posterior probability $1/|S_\tau|$ of having funded $\tau$. Were
$f_i(H)\ne f_j(H)$, strict monotonicity of $\mathbb{E}[q_a\mid H]$ in $f_a(H)$
would give the two accounts different posterior funding propensities, hence
different posterior probabilities of having funded $\tau$, and some
$a^\star\in S_\tau$ would satisfy $P(a^\star\mid H)>1/|S_\tau|$. Hence
$f_i(H)=f_j(H)$: incidence is constant on every masking set of $H$.

\end{proof}

The condition is local: it compares two accounts only when some transfer places
them in the same masking set. It becomes the global condition used throughout
the paper as soon as the masking sets are linked.

\begin{corollary}[Uniform incidence]
  \label{cor:uniform-incidence}
  Let $H$ be a history whose masking sets are \emph{connected under
  intersection}, meaning that their intersection graph is connected; this holds
  in particular when they pairwise intersect. Then an untraceable asset-transfer
  object satisfies strong untraceability on $H$ if and only if $f_i(H)=f_j(H)$
  for all $a_i,a_j\in\mathcal{A}(H)$.
\end{corollary}

\begin{proof}
If incidence is constant on $\mathcal{A}(H)$ it is constant on every masking set
of $H$, and Theorem~\ref{thm:untraceability-uniform} gives strong
untraceability. Conversely, let strong untraceability hold on $H$ and let
$a,b\in\mathcal{A}(H)$. The theorem makes incidence constant on each masking set
of $H$. Each of $a,b$ lies in some masking set, and by hypothesis there are
masking sets $S^{(1)},\dots,S^{(m)}$ of $H$ with $a\in S^{(1)}$, $b\in S^{(m)}$
and $S^{(l)}\cap S^{(l+1)}\ne\emptyset$ for each $l<m$. Consecutive sets share an
account at which their two constants agree, so one value runs along the whole
chain and $f_a(H)=f_b(H)$.
\end{proof}

Connectedness cannot be dropped. If $H$ uses two disjoint masking sets, one once
and the other twice, incidence is constant on each of them, so every transfer
keeps its members equiprobable and strong untraceability holds, yet incidence
takes two values on $\mathcal{A}(H)$.

The masking sets of every execution analyzed in Section~\ref{sec:ruat} pairwise
intersect (Lemma~\ref{lem:single-tau-upper}), so the corollary applies there
throughout.

The theorem quantifies over $\mathcal{A}(H)$, not over $\mathcal{A}$. Every
transfer creates previously unused account representations, and an account that has never
appeared in a masking set has incidence zero. Such an account is not a candidate
sender in $H$, and quantifying over it would leave strong untraceability
unsatisfiable. Restricted this way, the condition constrains the masking sets
alone, and is therefore indifferent to whether a transfer keeps its masking set
or replaces it. The two objects of the next sections differ in exactly that
respect: under the same condition, one keeps a consensus number of $2$ while the
other's grows quadratically in $\lambda$.

\section{Linear untraceable asset transfer}
\label{sec:luat}

Zcash~\cite{DBLP:conf/sp/Ben-SassonCG0MTV14} and Monero~\cite{zerotomonero2020}
retain spent accounts: they publish a growing \emph{allow-set} of potentially
unspent accounts together with a \emph{deny-set} of nullifiers marking the spent
ones, and a transfer proves in zero knowledge that some account of a masking set
has been nullified, without revealing which. The two systems differ in how a
transfer names that set. In Zcash it is implicit: the proof establishes
membership in the whole note-commitment tree, so the masking set is every
account on the ledger when the transfer is issued and $\lambda$ grows with it.
In Monero it is explicit: a ring signature names a small set of decoys chosen by
the sender, and $\lambda$ is the ring size, fixed by the protocol; full-chain
membership proofs~\cite{monerofcmp2024} would move Monero to the implicit form
as well. The distinction sets $\lambda$ but not the semantics: in both cases the
transfer names its masking set and consumes one account of it. We model this
family as the \emph{linear untraceable asset transfer} ($\luat$) object, whose
state grows linearly in the number of transfers.

Here we give the formal $\luat_\lambda$ model,
determine its consensus number for unshared and for shared accounts, show that
neither untraceability notion changes it, and analyze the garbage collection
forced by its state growth. Section~\ref{sec:ruat} then treats the
constant-state family, whose transfers consume and replace the whole masking
set.

\subsection{Formal model}
\label{subsec:luat-model}

The object maintains an \emph{allow-set} $\mathcal{A}_{\text{allow}}$ of all
potentially valid accounts and a \emph{deny-set} $\mathcal{N}$ of nullifiers for
spent accounts. Accounts are single-use: each transfer consumes one account and
creates new ones. To spend an account, a process proves knowledge of a secret
key for some account of a \emph{masking set} (a subset of the allow-set) without
revealing which, and the transfer publishes a nullifier that invalidates exactly
one account of that set. A \emph{nullifier} is a value
$\nu\in\mathcal{N}_{\text{space}}$ derived deterministically from an account via
a collision-resistant function
$\mathsf{Nullify}:\mathcal{A}\to\mathcal{N}_{\text{space}}$; distinct accounts
therefore have distinct nullifiers.

\begin{definition}[$\luat_\lambda$ object]
  \label{def:uat-lambda}
  Let $\mathcal{A}$ be a finite set of accounts and
  $\mu:\mathcal{A}\to 2^{\mathcal{P}}\setminus\{\emptyset\}$ the owner map. Fix
  $\lambda\ge 2$. The $\luat_\lambda$ object is the tuple
  $(Q,q_0,O,R,\Delta,U)$, where:
  \begin{itemize}
    \item \textbf{States:} $Q$ consists of triples
      $q=(\beta,\mathcal{A}_{\text{allow}},\mathcal{N})$ where
      $\mathcal{A}_{\text{allow}}\subseteq\mathcal{A}$ is the allow-set,
      $\beta:\mathcal{A}_{\text{allow}}\to\mathbb{N}$ assigns balances, and
      $\mathcal{N}\subseteq\mathcal{N}_{\text{space}}$ is the deny-set. The
      initial state is $q_0=(\beta_0,\mathcal{A}_0,\emptyset)$.

    \item \textbf{Operations:} $O=\{\mathsf{read}(a):a\in\mathcal{A}\}\cup
      \{\mathsf{transfer}(S,a_s,\nu,\mathcal{A}_{\text{new}},\beta_{\text{new}}):
      S\subseteq\mathcal{A},\,a_s\in S,\,\nu\in\mathcal{N}_{\text{space}},\,
      \mathcal{A}_{\text{new}}\subseteq\mathcal{A},\,
      \beta_{\text{new}}:\mathcal{A}_{\text{new}}\to\mathbb{N}\}$.

    \item \textbf{Responses:} $R=\{(\mathcal{A}_{\text{allow}},\mathcal{N})\}
      \cup\mathbb{N}\cup\{\bot\}$: a transfer returns the resulting allow-set
      and deny-set, and a read returns a balance or $\bot$.

    \item \textbf{Transitions:} For a state
      $q=(\beta,\mathcal{A}_{\text{allow}},\mathcal{N})$, a process
      $p\in\mathcal{P}$, an operation $o\in O$, a response $r\in R$, and a new
      state $q'=(\beta',\mathcal{A}'_{\text{allow}},\mathcal{N}')$, we have
      $(q,p,o,q',r)\in\Delta$ if and only if one of the following holds:
      \begin{itemize}
        \item $o=\mathsf{read}(a)\land q'=q\land
          r=\begin{cases}\beta(a) & \text{if } a\in\mathcal{A}_{\text{allow}}
          \text{ and } p\in\mu(a)\\ \bot &
          \text{otherwise}\end{cases}$;
        \item $o=\mathsf{transfer}(S,a_s,\nu,\mathcal{A}_{\text{new}},\beta_{\text{new}})\land
          S\subseteq\mathcal{A}_{\text{allow}}\land|S|\ge\lambda\land a_s\in S\land
          p\in\mu(a_s)\land\mathsf{Nullify}(a_s)=\nu\land\nu\notin\mathcal{N}\land
          \mathcal{A}_{\text{new}}\cap\mathcal{A}_{\text{allow}}=\emptyset\land
          \sum_{a\in\mathcal{A}_{\text{new}}}\beta_{\text{new}}(a)=\beta(a_s)\land
          (\forall a\in\mathcal{A}_{\text{new}}:\beta'(a)=\beta_{\text{new}}(a))\land
          (\forall a\in\mathcal{A}_{\text{allow}}:\beta'(a)=\beta(a))\land
          \mathcal{A}'_{\text{allow}}=\mathcal{A}_{\text{allow}}\cup\mathcal{A}_{\text{new}}\land
          \mathcal{N}'=\mathcal{N}\cup\{\nu\}\land
          r=(\mathcal{A}'_{\text{allow}},\mathcal{N}')$;
        \item otherwise, $q'=q\land r=(\mathcal{A}_{\text{allow}},\mathcal{N})$.
      \end{itemize}
    \item \textbf{Untraceability:} $U$ is the untraceability notion in force,
      giving $\luat_\lambda^-$ under Definition~\ref{def:luat-weak-untraceability}
      and $\luat_\lambda^+$ under Definition~\ref{def:luat-strong-untraceability}.
      As for the constant-state object, $U$ restricts the admissible executions.
  \end{itemize}
\end{definition}

A transfer $\mathsf{transfer}(S,a_s,\nu,\mathcal{A}_{\text{new}},\beta_{\text{new}})$
by $p$ succeeds exactly when $p$ owns $a_s$ inside the masking set $S$, the
nullifier $\nu$ is unused, and the new account representations have not been
used before. It then distributes
the balance of $a_s$ over $\mathcal{A}_{\text{new}}$ and adds $\nu$ to the
deny-set. We record three consequences of the definition.

First, invalidated accounts stay in $\mathcal{A}_{\text{allow}}$ with their
original balance. Removing them would reveal which account a nullifier
invalidated, so retaining them is what makes the allow-set monotone. Second, the
deny-set prevents double spending: once $\nu$ appears in $\mathcal{N}$, the
condition $\nu\notin\mathcal{N}$ fails in every later state. Third, two transfers
spending \emph{distinct} accounts do not interfere. Their nullifiers differ, the
allow-set only grows, and account-representation uniqueness makes their new
account sets disjoint,
so each succeeds irrespective of the other. Contention in $\luat$ thus requires
two processes to spend the \emph{same} account, hence an account that is shared.

\subsection{Consensus number}
\label{subsec:luat-cons-num}

We now determine the consensus number of $\luat_\lambda$. Protocols may use any
finite number of $\luat_\lambda$ instances together with atomic registers. The
analysis splits on ownership: Section~\ref{subsubsec:luat-unshared} treats
unshared objects, Section~\ref{subsubsec:luat-shared} objects with a $k$-shared
account, and Section~\ref{subsubsec:luat-untr-cons} determines the effect of the
untraceability notion in force, which we show to be none.

Both bounds follow from a single commutation property.

\begin{lemma}[Commuting transfers]
\label{lem:luat-commute}
Let $O$ be a $\luat_\lambda$ object in state $q$, and let $\tau_p$ and
$\tau_{p'}$ be transfers pending at distinct processes $p$ and $p'$ that spend
distinct accounts, so that $\nu_p\ne\nu_{p'}$. If both succeed when applied alone
at $q$, then both succeed in either order, and the two orders yield the same
state: applying $\tau_p$ then $\tau_{p'}$ at $q$ gives the same state as applying
$\tau_{p'}$ then $\tau_p$.
\end{lemma}

\begin{proof}
Suppose $\tau_p$ is applied first, taking $q=(\beta,\mathcal{A}_{\text{allow}},\mathcal{N})$
to $q_p=(\beta_p,\mathcal{A}_{\text{allow}}\cup\mathcal{A}_{\text{new},p},
\mathcal{N}\cup\{\nu_p\})$. We check each precondition of $\tau_{p'}$ at $q_p$.
Its masking set satisfies $S_{p'}\subseteq\mathcal{A}_{\text{allow}}\subseteq
\mathcal{A}_{\text{allow}}\cup\mathcal{A}_{\text{new},p}$, since the allow-set
only grows. Its nullifier satisfies $\nu_{p'}\notin\mathcal{N}\cup\{\nu_p\}$,
because $\nu_{p'}\notin\mathcal{N}$ by hypothesis and $\nu_{p'}\ne\nu_p$. Its new
accounts satisfy $\mathcal{A}_{\text{new},p'}\cap(\mathcal{A}_{\text{allow}}
\cup\mathcal{A}_{\text{new},p})=\emptyset$, since
$\mathcal{A}_{\text{new},p'}\cap\mathcal{A}_{\text{allow}}=\emptyset$ by
hypothesis and $\mathcal{A}_{\text{new},p'}\cap\mathcal{A}_{\text{new},p}=\emptyset$
by account-representation uniqueness. Finally
$\beta_p(a_s^{p'})=\beta(a_s^{p'})$, because
$\tau_p$ preserves balances on $\mathcal{A}_{\text{allow}}$, so the balance
condition $\sum_{a}\beta_{\text{new},p'}(a)=\beta(a_s^{p'})$ still holds. Hence
$\tau_{p'}$ succeeds at $q_p$, and symmetrically $\tau_p$ succeeds after
$\tau_{p'}$.

Applying both in either order yields the allow-set
$\mathcal{A}_{\text{allow}}\cup\mathcal{A}_{\text{new},p}\cup\mathcal{A}_{\text{new},p'}$,
the deny-set $\mathcal{N}\cup\{\nu_p,\nu_{p'}\}$, and the balance map that agrees
with $\beta$ on $\mathcal{A}_{\text{allow}}$, with $\beta_{\text{new},p}$ on
$\mathcal{A}_{\text{new},p}$, and with $\beta_{\text{new},p'}$ on
$\mathcal{A}_{\text{new},p'}$. Union is commutative and the three domains are
pairwise disjoint, so the resulting state does not depend on the order.
\end{proof}

\subsubsection{Unshared accounts}
\label{subsubsec:luat-unshared}

On an unshared object, every account has one owner, so any two processes
necessarily spend distinct accounts and Lemma~\ref{lem:luat-commute} applies to
every pair. The object nevertheless has non-trivial power: a process learns
\emph{how many} transfers preceded its own by reading the size of the deny-set.

\begin{lemma}[Lower bound]
  \label{lem:uat-lower-bound}
  For any $\lambda\ge 2$, an unshared $\luat_\lambda$ object implements
  wait-free $2$-process consensus. Hence $\consnum{\luat_\lambda}\ge 2$.
\end{lemma}

\begin{proof}
  Let $p_1,p_2$ have inputs $v_1,v_2$, let $a_1,a_2\in\mathcal{A}_0$ be accounts
  with $\mu(a_i)=\{p_i\}$, and let
  $S=\{a_1,a_2,a_3,\dots,a_\lambda\}\subseteq\mathcal{A}_0$ be a masking set of
  size $\lambda$ containing both. The protocol is given in
  Figure~\ref{fig:2-consensus-luat}.

  \begin{figure}[ht!]
  \begin{tcolorbox}[protocol={2-process consensus from LUAT}]
    \step{\textbf{Shared state:}}
    \begin{codelines}
    \item $R[1], R[2]$ atomic registers (proposals)
    \item A $\luat_\lambda$ object with accounts $\mathcal{A}_0 = \{a_1, \ldots, a_\lambda\}$
    \item Masking set $S = \{a_1, \ldots, a_\lambda\}$, where $\mu(a_i) = \{p_i\}$ for $i \in \{1,2\}$
    \end{codelines}

    \step{$\mathbf{propose}(v)$ at process $p_i$, $i \in \{1,2\}$:}
    \begin{codelines}[resume]
    \item $R[i].\mathsf{write}(v)$ \textit{// Publish proposal}
    \item $\nu_i \gets \mathsf{Nullify}(a_i)$ \textit{// Compute nullifier}
    \item $\mathcal{A}_{\text{new},i} \gets \{b_i\}$ \textit{// One previously unused account}
    \item $\beta_{\text{new},i}(b_i) \gets \beta_0(a_i)$ \textit{// Transfer full balance}
    \item $(\mathcal{A}_{\text{allow}}, \mathcal{N}) \gets \mathsf{transfer}(S, a_i, \nu_i, \mathcal{A}_{\text{new},i}, \beta_{\text{new},i})$
    \item \textbf{if} $|\mathcal{N}| = 1$ \textbf{then} \textit{// First to transfer}
    \item \quad \textbf{return} $v$ \textit{// Decide own value}
    \item \textbf{else} \textit{// Second to transfer ($|\mathcal{N}| = 2$)}
    \item \quad \textbf{return} $R[3-i].\mathsf{read}()$ \textit{// Decide other's value}
    \end{codelines}
  \end{tcolorbox}
  \caption{Wait-free $2$-process consensus from an unshared $\luat_\lambda$ object.}
  \label{fig:2-consensus-luat}
\end{figure}

  \emph{Agreement.} Both transfers succeed: $p_i$ owns $a_i$, the nullifiers
  $\nu_1\ne\nu_2$ are distinct and initially absent from $\mathcal{N}=\emptyset$,
  and the new account representations have not been used before. By
  linearizability they are ordered; let $p_i$
  come first. Its response carries $\mathcal{N}=\{\nu_i\}$, so $|\mathcal{N}|=1$
  and $p_i$ decides $v_i$. The response of $p_j$, $j\ne i$, carries
  $\mathcal{N}=\{\nu_i,\nu_j\}$, so $|\mathcal{N}|=2$ and $p_j$ reads $R[i]$,
  which $p_i$ wrote before invoking its transfer, and decides $v_i$. Both decide
  $v_i$.

  \emph{Validity.} The decision is $v_1$ or $v_2$, both inputs.

  \emph{Termination.} The protocol has no loops and every operation is
  wait-free.
\end{proof}

The protocol uses the only information a $\luat$ transfer reveals about
concurrent activity, namely its own position in the linearization order. By
Lemma~\ref{lem:luat-commute} it reveals nothing else, which yields a matching
upper bound at $3$ processes. Two commuting transfers leave the object in the
same state in either order, so the two orders are indistinguishable to any
process that did not take part in them, while the two participants do
distinguish them by reading different deny-set sizes. The contradiction
therefore needs a third process, which exists only when $n\ge 3$.

\begin{lemma}[Upper bound]
  \label{lem:uat-upper-bound}
  For any $\lambda\ge 2$, no wait-free $3$-process consensus protocol exists
  using unshared $\luat_\lambda$ objects and atomic registers. Hence
  $\consnum{\luat_\lambda}\le 2$ on unshared objects.
\end{lemma}

\begin{proof}
  Suppose for contradiction that $\Pi$ is such a protocol for a process set
  $\mathcal{P}$ with $n=|\mathcal{P}|\ge 3$. By the standard critical-state
  construction~\cite{DBLP:journals/toplas/Herlihy91}, $\Pi$ has a reachable
  critical configuration $C$: $C$ is bivalent and every successor $C\cdot p$ is
  univalent. Partition $\mathcal{P}=\mathcal{P}_0\sqcup\mathcal{P}_1$ by the
  valency of the successors, where $\mathcal{P}_v=\{p: C\cdot p \text{ is }
  v\text{-valent}\}$; both parts are non-empty because $C$ is bivalent. Fix
  $p\in\mathcal{P}_0$ and $q\in\mathcal{P}_1$.

  \emph{The two pending steps leave identical shared state in either order.}
At a critical configuration neither pending step is a register operation~\cite{DBLP:journals/toplas/Herlihy91}: a read leaves the shared state unchanged, writes to distinct registers commute, and of two writes to one register the second overwrites the first, leaving a configuration that every process except the overwritten writer cannot distinguish from the one in which that writer never took its step. Each case makes two oppositely valent successors indistinguishable to some process that decides in both. Both pending steps are therefore operations on $\luat_\lambda$ objects. If they are on distinct
  objects, they act on disjoint state and trivially commute. If either is a
  $\mathsf{read}$, or a transfer that fails, it leaves the state unchanged and
  again commutes. In the remaining case both are successful transfers on a common
  object; since the object is unshared, $p$ and $q$ own disjoint sets of accounts
  and so spend distinct accounts, and Lemma~\ref{lem:luat-commute} gives
  $\nu_p\ne\nu_q$, both succeeding in either order and reaching a common state.
  In every case the shared state at $C\cdot p\cdot q$ equals that at
  $C\cdot q\cdot p$.

  \emph{A spectator derives the contradiction.} Neither step writes a register,
  so the local state of every process other than $p$ and $q$ is also the same in
  $C\cdot p\cdot q$ and $C\cdot q\cdot p$. Since $n\ge 3$, choose
  $s\in\mathcal{P}\setminus\{p,q\}$ and run $s$ alone from each of the two
  configurations. By wait-freedom $s$ decides in both, and because the two
  configurations are indistinguishable to $s$, it decides the same value $v$ in
  both. But $C\cdot p$ is $0$-valent, so every extension of $C\cdot p\cdot q$
  decides $0$, forcing $v=0$; and $C\cdot q$ is $1$-valent, so every extension of
  $C\cdot q\cdot p$ decides $1$, forcing $v=1$. This contradiction shows no such
  $\Pi$ exists.
\end{proof}

\begin{theorem}[Consensus number of $\luat_\lambda$]
  \label{thm:uat-consensus-number}
  For any $\lambda\ge 2$, an unshared $\luat_\lambda$ object has consensus
  number exactly $2$, under either untraceability notion.
\end{theorem}

\begin{proof}
  Immediate from Lemmas~\ref{lem:uat-lower-bound} and~\ref{lem:uat-upper-bound};
  Corollary~\ref{cor:luat-untraceability-free} verifies that both bounds survive
  the restriction to executions admissible under each notion.
\end{proof}

The bound does not depend on $\lambda$: enlarging the masking set improves
privacy at no synchronization cost. The reason is that a $\luat$ masking set is
an argument to a zero-knowledge proof rather than a set of accounts the transfer
consumes, so two transfers may name overlapping masking sets and both commit.
Section~\ref{sec:ruat} treats the object in which this no longer holds.

\subsubsection{Shared accounts}
\label{subsubsec:luat-shared}

Contention in $\luat$ requires a shared account. When $k$ processes co-own
$a^\star$, each of them may spend it but only one succeeds, since their
transfers carry the same nullifier $\nu^\star$ and the deny-set admits it once.
This is the $\luat$ analogue of the $k$-shared asset transfer of Guerraoui
\etal~\cite{DBLP:journals/dc/GuerraouiKMPS22}.

\begin{lemma}[Lower bound]
  \label{lem:k-uat-lower-bound}
  Let $k\ge 2$ and let $\mu(a^\star)=\{p_1,\dots,p_k\}$. Then
  $\{p_1,\dots,p_k\}$ implement wait-free $k$-process consensus using the
  $\luat_\lambda$ object and atomic registers.
\end{lemma}

\begin{proof}
  Initialize $a^\star$ with balance $2k$. Each $p_i$ spends $a^\star$ with the
  common nullifier $\nu^\star$ and creates two accounts $b_i,c_i$, both owned by
  all $k$ processes, with balances $i$ and $2k-i$, respectively. Thus every
  transfer has the same public shape, while the smaller of the two confidential
  balances identifies the winner. The protocol is given in
  Figure~\ref{fig:k-consensus-uat}.

  \begin{figure}[ht!]
  \begin{tcolorbox}[protocol={$k$-process consensus from LUAT with a shared account}]
    \step{\textbf{Shared state:}}
    \begin{codelines}
    \item $R[1], \ldots, R[k]$ atomic registers (proposals)
    \item A $\luat_\lambda$ object with shared account $a^\star$, $\mu(a^\star) = \{p_1, \ldots, p_k\}$, $\beta_0(a^\star)=2k$
    \item Account pairs $(b_i,c_i)$ with $\mu(b_i)=\mu(c_i)=\{p_1,\ldots,p_k\}$ for every $i\in[k]$
    \item Masking set $S \ni a^\star$ with $|S| \ge \lambda$
    \end{codelines}

    \step{$\mathbf{propose}(v)$ at process $p_i$, $i \in [k]$:}
    \begin{codelines}[resume]
    \item $R[i].\mathsf{write}(v)$ \textit{// Publish proposal}
    \item $\nu^\star \gets \mathsf{Nullify}(a^\star)$ \textit{// Shared nullifier}
    \item $\mathcal{A}_{\text{new},i} \gets \{b_i,c_i\}$ \textit{// Two previously unused accounts}
    \item $\beta_{\text{new},i}(b_i) \gets i$; $\beta_{\text{new},i}(c_i) \gets 2k-i$ \textit{// Encode $i$ in the balances}
    \item $(\mathcal{A}_{\text{allow}}, \mathcal{N}) \gets \mathsf{transfer}(S, a^\star, \nu^\star, \mathcal{A}_{\text{new},i}, \beta_{\text{new},i})$
    \item $\{b,c\} \gets \mathcal{A}_{\text{allow}}\setminus\mathcal{A}_0$ \textit{// The winner's new accounts}
    \item $\ell \gets \min\{\mathsf{read}(b),\mathsf{read}(c)\}$ \textit{// Decode the winner}
    \item \textbf{if} $\ell = i$ \textbf{then}
    \item \quad \textbf{return} $v$ \textit{// Won}
    \item \textbf{else} \textit{// Lost; the winner is $p_\ell$}
    \item \quad \textbf{return} $R[\ell].\mathsf{read}()$
    \end{codelines}
  \end{tcolorbox}
  \caption{Wait-free $k$-process consensus from a $\luat_\lambda$ object with a
  $k$-shared account. Every contender creates two accounts and encodes its
  identifier as the smaller of the two resulting balances.}
\label{fig:k-consensus-uat}
  \end{figure}

  \emph{Agreement.} By linearizability the $k$ transfers are ordered; let $p_j$
  come first. At its linearization point $a^\star\in S\subseteq
  \mathcal{A}_{\text{allow}}$, $\nu^\star\notin\mathcal{N}=\emptyset$,
  $p_j\in\mu(a^\star)$, and $\mathcal{A}_{\text{new},j}$ contains only
  previously unused representations, so $\tau_j$
  succeeds. The balance condition holds because
  $\beta_{\text{new},j}(b_j)+\beta_{\text{new},j}(c_j)=j+(2k-j)=2k=
  \beta_0(a^\star)$. It adds $\nu^\star$ to $\mathcal{N}$ and $b_j,c_j$ to the
  allow-set, so $\mathcal{A}_{\text{allow}}\setminus\mathcal{A}_0=\{b_j,c_j\}$.
  Every process owns both accounts and can read their balances. Since
  $j\le k\le 2k-j$, each obtains $\ell=\min\{j,2k-j\}=j$.

  Every later transfer $\tau_i$, $i\ne j$, finds $\nu^\star\in\mathcal{N}$ and
  fails, leaving the state unchanged. Hence $p_i$ also reads $\ell=j$. As
  $\ell=j\ne i$, it reads $R[j]$ and decides $v_j$, which $p_j$ wrote before
  invoking its transfer. Process $p_j$ computes $\ell=j=j$ and decides its own
  $v_j$. All decide $v_j$.

  \emph{Validity.} The decision is $R[j]$ for some $j\in[k]$, an input.

  \emph{Termination.} The protocol has no loops and every operation is
  wait-free.
\end{proof}

The balance encoding is available precisely because the source and the newly
created accounts are shared: every participant is authorized to read the
resulting balance. It cannot be used to obtain consensus from unshared accounts,
or in the CUAT constructions below, because balance confidentiality prevents a
losing process from reading an account owned by the winner. Those constructions
must therefore expose the winner through the object's synchronization pattern
rather than through balances.

Sharing lifts the bound to $k$ but no further. A process outside $\mu(a^\star)$
cannot contend for $\nu^\star$, so by Lemma~\ref{lem:luat-commute} its transfers
commute with those of every other process.

\begin{lemma}[Upper bound]
  \label{lem:k-uat-upper-bound}
  Let $k\ge 2$, let $a^\star$ be the only shared account of a $\luat_\lambda$
  object with $\mu(a^\star)=\{p_1,\dots,p_k\}$, and let $p_{k+1}\notin
  \mu(a^\star)$. Then $\{p_1,\dots,p_{k+1}\}$ have no wait-free
  $(k+1)$-process consensus protocol using that object and atomic registers.
\end{lemma}

\begin{proof}
  Suppose such a protocol $\Pi$ exists. As in
  Lemma~\ref{lem:uat-upper-bound}, take a reachable critical configuration $C$,
  partition $\mathcal{P}=\mathcal{P}_0\sqcup\mathcal{P}_1$ by successor valency
  with both parts non-empty, and note every pending step is an operation on the
  object. Process $p_{k+1}$ lies in one part; say $p_{k+1}\in\mathcal{P}_0$, the
  other case being symmetric. Choose $q\in\mathcal{P}_1$, so
  $q\in\{p_1,\dots,p_k\}$.

  Since $a^\star$ is the only shared account and $p_{k+1}\notin\mu(a^\star)$,
  every account $p_{k+1}$ may spend is owned by $p_{k+1}$ alone, so the pending
  steps of $p_{k+1}$ and $q$ spend distinct accounts. By
  Lemma~\ref{lem:luat-commute} (or trivially, if either step leaves the state
  unchanged), the shared state at $C\cdot p_{k+1}\cdot q$ equals that at
  $C\cdot q\cdot p_{k+1}$, and neither step writes a register.

  It remains to find a spectator. The process set has $k+1\ge 3$ members, so
  there is some $s\notin\{p_{k+1},q\}$. Running $s$ alone from the two
  configurations, wait-freedom makes it decide, and indistinguishability makes it
  decide the same value in both, contradicting the opposite valencies of
  $C\cdot p_{k+1}$ and $C\cdot q$ exactly as in
  Lemma~\ref{lem:uat-upper-bound}.
\end{proof}

\begin{theorem}[Consensus number of $k$-shared $\luat_\lambda$]
  \label{thm:k-uat-consensus-number}
  If $k$ processes share ownership of an account $a^\star$ in an otherwise
  unshared $\luat_\lambda$ object, the consensus number among them is
  $\max(k,2)$, under either untraceability notion.
\end{theorem}

\begin{proof}
  For $k\ge 2$, Lemmas~\ref{lem:k-uat-lower-bound}
  and~\ref{lem:k-uat-upper-bound} give matching bounds of $k$. For $k=1$ the
  object is unshared and Theorem~\ref{thm:uat-consensus-number} gives $2$.
  Corollary~\ref{cor:luat-untraceability-free} verifies both bounds under each
  notion.
\end{proof}

The floor at $2$ distinguishes $\luat$ from standard asset
transfer~\cite{DBLP:journals/dc/GuerraouiKMPS22}, whose $k$-shared consensus
number is exactly $k$ and drops to $1$ without sharing. It is due to the
deny-set: even with no shared account, a $\luat$ transfer reports how many
transfers preceded it, which a plain asset transfer does not. This is consistent
with Frey \etal~\cite{DBLP:conf/wdag/FreyGR23}, who locate the synchronization
cost of private payments in the deny-list rather than the allow-list.

\subsubsection{Untraceability and the consensus number}
\label{subsubsec:luat-untr-cons}

Neither bound above mentions the untraceability notion in force. Strong
untraceability does constrain how $\luat$ may be used, by the same mechanism as
in the constant-state case: by Corollary~\ref{cor:uniform-incidence},
masking-set incidence must be uniform. Uniformity is therefore a real constraint
on $\luat$: histories that over-use an
account are inadmissible under $\luat_\lambda^+$, and the hub families of
Section~\ref{subsubsec:weak-untr} are excluded here too. It does not, however,
affect the consensus number, and the following corollary makes the two bounds of
Sections~\ref{subsubsec:luat-unshared} and~\ref{subsubsec:luat-shared} explicit
in the untraceability notion.

\begin{corollary}[Untraceability does not affect $\consnum{\luat_\lambda}$]
  \label{cor:luat-untraceability-free}
  For every $\lambda\ge 2$,
  \[
    \consnum{\luat_\lambda^-}=\consnum{\luat_\lambda^+}=
    \begin{cases}
      2 & \text{on unshared objects,}\\
      \max(k,2) & \text{on objects with a $k$-shared account.}
    \end{cases}
  \]
\end{corollary}

\begin{proof}
  The untraceability notion in force restricts the admissible executions, hence
  the masking sets a protocol may carry; we check that neither bound depends on
  that restriction.

  \emph{Upper bounds.} Lemmas~\ref{lem:uat-upper-bound}
  and~\ref{lem:k-uat-upper-bound} rest on Lemma~\ref{lem:luat-commute}, whose
  only hypothesis is that the two transfers spend distinct accounts. Masking sets
  are not mentioned, so the bounds hold over every family of masking sets, hence
  over the admissible executions of either variant. Restricting the executions
  can only lower a consensus number, so both variants are bounded as stated.

  \emph{Lower bounds.} It suffices to exhibit admissible executions of
  $\luat_\lambda^+$ realizing them, strong untraceability implying weak
  untraceability as noted in Section~\ref{sec:untraceability}. The
  protocols of Figures~\ref{fig:2-consensus-luat} and~\ref{fig:k-consensus-uat}
  fix one masking set $S$ with $|S|\ge\lambda$, used by every invocation, and
  constrain it no further. Take $\mathcal{A}_0=S$. Every transfer of the
  execution, successful or not, carries $S$, so each $a\in S$ has incidence
  $f_a(H)=|H|$, while the accounts created during the execution lie in no masking
  set and are excluded from $\mathcal{A}(H)$. Incidence is therefore constant on
  $\mathcal{A}(H)=S$, and Theorem~\ref{thm:luat-untraceability-uniform} makes the
  execution admissible under strong untraceability. The decisions are unaffected,
  so the protocols achieve $2$- and $k$-consensus respectively on
  $\luat_\lambda^+$.
\end{proof}

The corollary separates two things that coincide for $\cuat_\lambda$. Uniformity
constrains which masking sets a protocol may use; it does not constrain
synchronization, because in $\luat$ the masking set is inert and overlapping
masking sets do not conflict. A protocol is therefore free to satisfy
untraceability with one choice of masking sets and to obtain its synchronization
from a disjoint mechanism, namely the shared account. Section~\ref{sec:ruat}
treats an object whose transfers consume their entire masking set; there the two
coincide, since overlap is conflict, and the same uniformity constraint
determines the consensus number.

\subsection{State growth and garbage collection}
\label{subsec:luat-gc}

The counterpart of $\luat$'s low synchronization cost is its state. Both the
allow-set and the deny-set grow monotonically: every transfer adds a nullifier to
$\mathcal{N}$ and previously unused account representations to
$\mathcal{A}_{\text{allow}}$, while retaining
the invalidated ones. The state therefore scales with the number of transfers
rather than with the number of active accounts, and neither structure can be
pruned without sacrificing untraceability (removing an invalidated account
reveals which nullifier invalidated it) or security (dropping a nullifier permits
a double spend).

One remedy partitions accounts into pre-declared, immutable masking sets. Once a
masking set of size $k$ has been used $k$ times, its accounts are exhausted and
both the set and its nullifiers can be collected. This technique, studied by
Chow \etal~\cite{DBLP:conf/csfw/ChowELRW23} and Cachin and
Wicht~\cite{DBLP:journals/popets/CachinW25}, uses an auxiliary
\emph{partitioning object} maintaining a consistent partition of the account set
into parts of bounded size. We formalize it and determine its consensus number,
which settles whether garbage collection costs more synchronization than $\luat$
itself. Collection need not be triggered by exhaustion:
Anonymous Zether~\cite{DBLP:conf/fc/BunzAZB20} scopes each nullifier to an epoch and
discards the deny-set at every epoch boundary, which bounds the state at the
price of one transfer per account per epoch. The two triggers differ in progress
as well as in state. Exhaustion is caused by the transfers themselves, so a
transfer that is eligible when invoked stays eligible and
Theorem~\ref{thm:luat-fairness} still applies; an epoch boundary is
external to them, and a transfer whose proof is not included before it expires
must be reissued, which makes progress depend on a synchrony assumption that
$\luat_\lambda$ itself does not need.

\begin{definition}[$k$-partitioning object]
  \label{def:partitioning}
  Let $\mathcal{A}$ be a set of accounts, $\prec$ a fixed total order on
  $\mathcal{A}$, and $k\ge 1$. The \emph{$k$-partitioning object} is the tuple
  $(Q,q_0,O,R,\Delta)$ where:
  \begin{itemize}
    \item \textbf{States:} $Q$ consists of the families
      $\mathcal{B}\subseteq 2^{\mathcal{A}}$ of pairwise-disjoint non-empty
      \emph{parts} with $|B|\le k$ for every $B\in\mathcal{B}$, of which at
      most one is \emph{incomplete} (of size $<k$). The initial state is
      $q_0=\emptyset$.
    \item \textbf{Operations and responses:}
      $O=\{\mathsf{partition}(A):A\subseteq\mathcal{A}\}$ and
      $R=2^{2^{\mathcal{A}}}$.
    \item \textbf{Transitions:} On $\mathsf{partition}(A)$ in state
      $\mathcal{B}$, let $U=A\setminus\bigcup\mathcal{B}$ be the yet
      unpartitioned elements of $A$. The new state $\mathcal{B}'$ is obtained
      by:
      \begin{enumerate}[(i)]
        \item if $\mathcal{B}$ has an incomplete part $B^*$, replacing it by
          $B^*\cup A'$, where $A'$ consists of the
          $\min(k-|B^*|,|U|)$ smallest elements of $U$ under $\prec$;
        \item partitioning $U\setminus A'$, taken in $\prec$-order, into
          consecutive parts of size $k$, with at most one final part of size
          $<k$.
      \end{enumerate}
      The response is $r=\{B\cap A: B\in\mathcal{B}',\,B\cap A\ne\emptyset\}$.
  \end{itemize}
  An execution is \emph{admissible} if the sets partitioned by distinct
  operations are pairwise disjoint and each is partitioned by its owner. This
  matches the garbage-collection setting, where a process partitions the
  accounts it has just created.
\end{definition}

Fixing $\prec$ makes the transition deterministic, so the object has a
well-defined sequential specification; the invariant that at most one part is
incomplete is preserved, since step~(i) consumes the only incomplete part and
step~(ii) leaves at most one. The response reports the partition
\emph{restricted} to $A$, so a process learns how its own accounts were grouped
and nothing else. This withholds nothing the application needs. A process does
eventually require the whole part, to name its co-members as decoys, but it
reads them from the ledger once the part is complete, and a register read has
consensus number $1$. What the restriction denies is only that this information
arrive \emph{atomically} with the operation, and that is where the
synchronization power would lie: were the response to expose whole parts, a
process could read foreign accounts out of its own part and reconstruct the
arrival order of other operations. Already at $k=3$, three processes each
partitioning a pair $\{x_i,y_i\}$ would identify the first to arrive, raising the
consensus number to at least $3$. Collection itself needs no coordination and is
therefore excluded from the object: once the partition is agreed, a part is
discarded when its own $k$ accounts have been spent, which every process can
determine locally.

Two conditions on $\prec$ matter. It must be fixed in advance and identical at
every process, evaluated locally on account representations and never negotiated
during the execution, since agreeing on an order would itself be an agreement
problem; this is what keeps $\prec$ out of the synchronization analysis, and
Lemma~\ref{lem:part-faa} confirms it by implementing the object from a single
fetch-and-add. It must also be independent of ownership, so that membership in a
part does not follow from who created an account. Independence is necessary but
not sufficient. Under admissibility a process partitions accounts it has just
created, and step~(ii) forms parts from those accounts alone, so an invocation
with $|A|\ge k$ fills a part with one owner's accounts whatever $\prec$ is; only
the part inherited by step~(i) mixes owners. The parts therefore serve as
masking sets when invocations are small relative to $k$, which is the
garbage-collection regime, where a transfer creates an output and a change
account. The degenerate case is invisible to
Theorem~\ref{thm:untraceability-uniform}: every account lies in exactly one
part, so incidence is uniformly $1$ and the funding account remains one of $k$
candidates, but those candidates share an owner and the transfer identifies its
sender.

The object reduces to a standard primitive.

\begin{lemma}[Reduction to fetch-and-add]
  \label{lem:part-faa}
  Let $t=|\bigcup\mathcal{B}|$ denote the number of partitioned elements before
  an operation. Along admissible executions:
  \begin{enumerate}[(i)]
    \item the $k$-partitioning object is implementable from a fetch-and-add
      register modulo $k$ and atomic registers; and
    \item writing $B$ for the response part containing the $\prec$-least
      element of $A$, we have $|B|=\min\big(k-(t\bmod k),\,|A|\big)$.
  \end{enumerate}
\end{lemma}

\begin{proof}
  Admissibility gives $U=A$ for every operation, since no element of $A$ has
  been partitioned before. By the state invariant at most one part is
  incomplete, and it has size $t\bmod k$: the partitioned elements fill parts
  of size $k$ in order, leaving $t\bmod k$ elements in the last one. Steps~(i)
  and~(ii) therefore place the elements of $A$, in $\prec$-order, into the
  $k-(t\bmod k)$ free slots of the incomplete part and then into new parts
  of size $k$. The induced partition of $A$, which is exactly the response, is
  thus a function of $t\bmod k$, $|A|$, and the $\prec$-order of $A$; the last
  two are local to the invoking process. The effect on the shared state is
  $t\mapsto t+|A|$.

  This yields claim~(i): on $\mathsf{partition}(A)$, fetch-and-add $|A|$ to a
  register modulo $k$, obtaining $t\bmod k$, and compute the response locally by
  the rule above. Each operation takes a single fetch-and-add, so the
  implementation is wait-free and linearizes at that step.

  For claim~(ii), write $c=t\bmod k$ and let $a_1$ be the $\prec$-least element
  of $A$. If $c=0$ there is no incomplete part, so $a_1$ opens a new part,
  which receives the $\prec$-first $\min(k,|A|)$ elements of $A$. If $c>0$ the
  incomplete part has $k-c$ free slots, and $a_1$ joins it along with the
  $\prec$-next $\min(k-c,|A|)-1$ elements of $A$. Admissibility makes $A$
  disjoint from every other partitioned set, so intersecting with $A$ removes
  exactly the foreign elements of that part, leaving $|B|=\min(k-c,|A|)$. Both
  cases agree with $|B|=\min\big(k-(t\bmod k),|A|\big)$, since $c=0$ gives
  $k-c=k$.
\end{proof}

\begin{theorem}[Consensus number of $k$-partitioning]
  \label{thm:part-consensus-number}
  Along admissible executions, $\consnum{1\text{-}\mathrm{Part}}=1$ and
  $\consnum{k\text{-}\mathrm{Part}}=2$ for every $k\ge 2$.
\end{theorem}

\begin{proof}
  \emph{Case $k=1$.} Here $t\bmod 1=0$ always, so by Lemma~\ref{lem:part-faa}
  the response is the constant $\{\{a\}:a\in A\}$, which the invoking process
  computes without consulting the shared state. The object is implementable from
  nothing, so it adds no power to atomic registers, whose consensus number is
  $1$~\cite{DBLP:journals/toplas/Herlihy91}. Hence
  $\consnum{1\text{-}\mathrm{Part}}=1$.

  \emph{Upper bound for $k\ge 2$.} By Lemma~\ref{lem:part-faa} the object is
  implementable from a fetch-and-add register and atomic registers.
  Fetch-and-add has consensus number
  $2$~\cite{DBLP:journals/toplas/Herlihy91}, and an object implementable from
  objects of consensus number $2$ cannot exceed $2$, since otherwise the
  implementation would solve $3$-process consensus from consensus-number-$2$
  objects. Hence $\consnum{k\text{-}\mathrm{Part}}\le 2$. The bound is not an
  artefact of the reduction. By claim~(ii) each of $d$ processes partitioning
  $k+1$ accounts sees a different response size, and so learns its own position
  in the linearization order; for $k\ge 3$ the first three see $k$, $k-1$ and
  $k-2$. A position is not an identity, however: the third process learns that it
  is third but not which of the other two preceded it, and so has no register it
  can adopt. Only at two processes does a position determine the other party,
  which is the same reason fetch-and-add itself stops at $2$.

  \emph{Lower bound for $k\ge 2$.} Let $p_1,p_2$ hold disjoint account sets
  $A_1,A_2$ with $|A_i|=k+1$ and inputs $v_1,v_2$. Each $p_i$ writes $v_i$ to
  the single-writer register $R[i]$, invokes $\mathsf{partition}(A_i)$, and
  inspects the part $B_i$ of its response that contains the $\prec$-least
  element of $A_i$. It decides its own input if $|B_i|=k$, and otherwise reads
  $R[3-i]$ and decides that value.

  By linearizability the two operations are ordered; let $p_i$ come first. It
  meets $t=0$, so claim~(ii) of Lemma~\ref{lem:part-faa} gives
  $|B_i|=\min(k,k+1)=k$, and $p_i$ decides $v_i$. Process $p_j$, $j\ne i$, meets
  $t=|A_i|=k+1$, hence $t\bmod k=1$ and $|B_j|=\min(k-1,k+1)=k-1$. As
  $k-1\ne k$, process $p_j$ takes the second branch and reads $R[i]$, which
  $p_i$ wrote before invoking its operation, and decides $v_i$. Both decide
  $v_i$, so agreement holds; validity is immediate, as the decision is $v_1$ or
  $v_2$; and the protocol has no loops, so it is wait-free. Hence
  $\consnum{k\text{-}\mathrm{Part}}\ge 2$.

  The choice $|A_i|=k+1$ matters: it is the smallest size that is neither a
  multiple of $k$ (which would leave $t\bmod k$ at $0$ for both processes, making
  the responses coincide) nor equal to $1$ (for which the response is the
  singleton $\{A_i\}$ in every state, carrying no information).
\end{proof}

For $k\ge 2$ the partitioning object therefore has the same consensus number as
$\luat$ itself, so partitioning-based garbage collection adds no synchronization
requirement to the classical design.

Randomizing the assignment does not change this number. Cachin and
Wicht~\cite{DBLP:journals/popets/CachinW25} draw the partition from a
distributed randomness beacon, so that an adversary cannot grind account keys to
steer its accounts into a chosen part; a public deterministic rule, including
the order $\prec$ above, offers no such guarantee. The beacon value is read-only
shared data, of consensus number $1$, and it selects which part an account
joins, not how many free slots that part still has. Enforcing $|B|\le k$
requires the object to track that occupancy, so the response continues to report
how much of a part was already filled, which is exactly what the lower bound
uses. Randomization therefore leaves the object at $2$; what it buys is grinding
resistance and the independence from ownership discussed above, and what it
costs is a stronger failure assumption. Unpredictability requires the beacon
value to be unknown when the accounts are created, so a partition waits for the
beacon to publish, and publishing needs enough correct processes to contribute,
a threshold set by the beacon scheme. The partitioning object is wait-free and
tolerates any number of crashes, whereas the randomness it consumes is available
only under a bound on the number of faulty processes. This differs from the
epoch-scoped nullifiers noted above, which instead presuppose synchrony. A
beacon that must also agree on which contributions were counted adds
synchronization of its own, which a beacon whose round value is a unique
threshold signature on the round number avoids.

To summarize the linear-state family: its consensus number is $2$, independently
of $\lambda$ and of the untraceability notion in force, and $\max(k,2)$ under
$k$-sharing; strong untraceability constrains its masking sets but not its
synchronization; and its state growth is collectable at no synchronization cost.
Its cost is storage. Section~\ref{sec:ruat} removes that cost and determines what
replaces it.

\section{Constant untraceable asset transfer}
\label{sec:ruat}

\emph{Quisquis} and similar systems~\cite{DBLP:conf/asiacrypt/FauziMMO19,DBLP:journals/iacr/AlupothaGC24,DBLP:conf/asiacrypt/MadathilS25} address the
monotonic growth that Section~\ref{subsec:luat-gc} charged to $\luat$, and they
do so by discarding the spent account instead of retaining it: each transfer
invalidates \emph{all} accounts of its masking set and creates new ones in their
place, while cryptographic proofs hide which of them moved value. The ledger
then stores only current accounts, and its size is bounded by the number of
active accounts rather than by the number of transfers.

We model this family as the \emph{constant untraceable asset transfer} (CUAT)
object, again building on the asset transfer formalism of Guerraoui
\etal~\cite{DBLP:journals/dc/GuerraouiKMPS22}. The common untraceability notions
and their uniform-incidence characterization appear in
Section~\ref{sec:untraceability}; this section gives the formal CUAT model and
determines its consensus number.

One difference from Section~\ref{sec:luat} accounts for the difference in the
results. A $\luat$ transfer consumes a single account and names the rest of its
masking set as decoys, so overlapping masking sets do not conflict and, by
Lemma~\ref{lem:luat-commute}, transfers on distinct accounts commute. A CUAT
transfer consumes its masking set. Two transfers whose masking sets meet cannot
both commit, so overlap is conflict, the masking-set family carries a graph
structure, and the uniformity required by strong untraceability, which does not
affect $\luat$ (Corollary~\ref{cor:luat-untraceability-free}), bears on the
consensus number.

\subsection{Formal model}
\label{subsec:ruat-model}
We now formalize the $\lambda$-CUAT object. Each transfer atomically invalidates an old masking set and creates a new one of size $\lambda$, with a zero-knowledge proof witnessing ownership of the invalidated accounts. In the terms of Definition~\ref{def:masking-set}, $\cuat_\lambda$ is $\lambda$-masking, whereas $\luat_\lambda$ only requires at least $\lambda$ members.
\begin{definition}[$\cuat_\lambda$ object]
  \label{def:ruat-lambda}
  $\cuat_\lambda$ is a shared-memory asset
  transfer object specified by the
  tuple $\cuat_\lambda=(Q,q_0,O,R,\Delta, U)$ where:
  \begin{itemize}
    \item \textbf{States:} $Q$ consists of pairs $q=(\mathcal{A},\beta)$ where
      $\mathcal{A}$ is a finite set of currently valid accounts and
      $\beta:\mathcal{A}\to\mathbb{N}$ maps each account to its balance.
      Each account $a\in\mathcal{A}$ carries two fixed fields assigned
      at creation: an owner set $\mu(a)\subseteq\mathcal{P}$ and an index
      $\mathrm{idx}(a)\in\mathbb{N}$ also denoted by a subscript.
      The initial state is $q_0=(\mathcal{A}_0,\beta_0)$ for some
      finite $\mathcal{A}_0$.

    \item \textbf{Operations:} The set $O$ consists of:
      $\rd{a}$ for any account $a$, and
      $\tr{S_{\beta'},\pi}$ where
      $S_{\beta'} = \{a'_1, \ldots, a'_\lambda\}$ is a set of $\lambda$ new accounts
      and $\pi$ is a zero-knowledge proof.
      The transfer carries an implicit bijection $\phi: S_{\beta} \to S_{\beta'}$
      between the old and new masking set that preserves both account
      fields: $\mu(\phi(a))=\mu(a)$ and $\mathrm{idx}(\phi(a))=\mathrm{idx}(a)$
      for every $a\in S_{\beta}$.

    \item \textbf{Responses:} $R$ consists of:
      \begin{itemize}
        \item For $\rd{a}$: either $\beta(a)$ if the invoking process
          owns $a$, or $\bot$ otherwise.
        \item For $\tr{S_{\beta'},\pi}$: a pair $(\mathcal{A}, b)$ where $\mathcal{A}$ is the current account set
          and $b \in \{\true, \false\}$ indicates success or failure.
      \end{itemize}

    \item \textbf{Transitions:} $\Delta\subseteq Q\times O\times R\times Q$
      specifies allowed state transitions. For any $q=(\mathcal{A},\beta)$
      and $q'=(\mathcal{A}',\beta')$:
      \begin{itemize}
        \item \textbf{Account read:} If $o=\rd{a}$ invoked by process
          $p$, then $q'=q$ and
          \[
            r =
            \begin{cases}
              \beta(a) & \text{if } a \in \mathcal{A} \text{ and } p\in\mu(a), \\
              \bot & \text{otherwise.}
            \end{cases}
          \]

        \item \textbf{Transfer success:} If
            $o=\tr{S_{\beta'},\pi}$ with $S_{\beta} \subseteq \mathcal{A}$,
          $S_{\beta'} \cap \mathcal{A} = \emptyset$,
          $|S_{\beta'}| = \lambda$, and $\pi$ is a proof demonstrating
          $\exists p\in\mathcal{P}, a_s, a_r \in S_{\beta}, v \in \mathbb{N}$ such that
          $p \in \mu(a_s)$ and $\beta(a_s) \geq v$, then the new account
          set is
          \[
              \mathcal{A}' = (\mathcal{A} \setminus S_{\beta}) \cup S_{\beta'},
          \]
          and the new balance map $\beta':\mathcal{A}'\to\mathbb{N}$ is
          \[
            \beta'(a) =
            \begin{cases}
              \beta(\phi^{-1}(a)) - v & \text{if } a = \phi(a_s), \\
              \beta(\phi^{-1}(a)) + v & \text{if } a = \phi(a_r), \\
              \beta(\phi^{-1}(a))     & \text{if } a \in S_{\beta'} \setminus \{\phi(a_s),\phi(a_r)\}, \\
              \beta(a)                & \text{if } a \in \mathcal{A} \setminus S_{\beta}.
            \end{cases}
          \]
          Return $r=(\mathcal{A}', \true)$.
        \item \textbf{Transfer failure:} In all other cases of
          $\tr{\cdot}$, set $q'=q$ and $r=(\mathcal{A}, \false)$.
      \end{itemize}
    \item \textbf{Untraceability:} $U$ specifies the untraceability notion required for the object. We consider two notions: weak
      untraceability (Definition~\ref{def:weak-untraceability}), giving the object $\cuat_\lambda^-$, and strong untraceability (Definition~\ref{def:strong-untraceability}), giving $\cuat_\lambda^+$. We write $\cuat_\lambda$ for either variant when the result or argument applies to both. The notion restricts the valid executions of the object, as some use patterns lead to untraceability loss, especially under strong untraceability.
  \end{itemize}
\end{definition}

The $\lambda$-CUAT models privacy-preserving cryptocurrencies
through account invalidation and creation. A transfer operation takes
a set $S_{\beta'} = \{a'_1, \ldots, a'_\lambda\}$ of $\lambda$ new accounts
and a zero-knowledge proof $\pi$.
The proof $\pi$ demonstrates ownership and valid balance updates
without revealing the sender or recipient.
Successful transfers atomically invalidate all accounts in $S_{\beta}$ and
create all accounts in $S_{\beta'}$ with updated balances.
This atomic masking-set update is the key mechanism: all old accounts are removed,
all new accounts are created, and the cryptographic proof hides which accounts
are the sender and recipient.
Balance confidentiality ensures that only
account owners can observe account balances.
Figure~\ref{fig:randomization} illustrates the effect of a successful
transfer on the masking set.
\begin{example}[Transfer on a $3$-element masking set]
  \label{ex:transfer-3-set}
  Consider a system with $\lambda=3$. Suppose the current state has
  $\mathcal{A} = \{a_1, \ldots, a_8\}$ with accounts $a_3, a_5, a_7
  \in \mathcal{A}$. Process $p_3$
  (owner of $a_3$) invokes $\tau=\tr{S_{\beta'},\pi}$ with $S_{\beta'} = \{a_3', a_5',
  a_7'\}$ (new accounts) where $S_{\beta} = \{a_3, a_5, a_7\}$
  is paired with $S_{\beta'}$ via the field-preserving bijection $\phi: a_i \mapsto a_i'$
  (so $\mu(a_i')=\mu(a_i)$ and $\mathrm{idx}(a_i')=\mathrm{idx}(a_i)=i$).
  The proof $\pi$ demonstrates that $p_3$ owns some account in $S_{\beta}$
  and that it transfers $v$ units between two accounts
  while preserving all other balances.
  If $\tau$ succeeds, the state transition yields:
Accounts: $\mathcal{A}' = (\mathcal{A} \setminus \{a_3,
      a_5, a_7\}) \cup \{a_3', a_5', a_7'\}$.
The zero-knowledge proof $\pi$ hides which accounts in
  $S_{\beta} = \{a_3,a_5,a_7\}$ are the sender and recipient.
  Any concurrent transfer using an overlapping masking set will fail
  after $\tau$ completes, since the old accounts $a_3, a_5, a_7$ are
  no longer valid. Accounts $a_i$ and $a_i'$ share the same index field ($\mathrm{idx}(a_i)=\mathrm{idx}(a_i')=i$), representing the same logical position despite differing cryptographic representations.
\end{example}
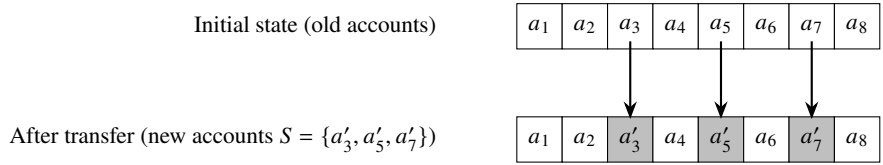
\begin{figure}[ht!]
  \centering
  \begin{tikzpicture}[scale=0.75, every node/.style={scale=0.8}]
\def\n{8}           \def\cellwidth{0.8} \foreach \idx in {1,...,\n} {
      \draw[fill=white] (\idx*\cellwidth, 0) rectangle ++(\cellwidth,
      \cellwidth);
      \node (a\idx) at (\idx*\cellwidth + 0.5*\cellwidth, 0.5*\cellwidth)
      {$a_{\idx}$};
    }
\foreach \i in {1,...,\n} {
      \ifnum\i=3
      \draw[fill=gray!50] (\i*\cellwidth, -2) rectangle
      ++(\cellwidth, \cellwidth);
      \node (a\i') at (\i*\cellwidth + 0.5*\cellwidth, -2 +
      0.5*\cellwidth) {$a_{\i}'$};
      \else\ifnum\i=5
      \draw[fill=gray!50] (\i*\cellwidth, -2) rectangle
      ++(\cellwidth, \cellwidth);
      \node (a\i') at (\i*\cellwidth + 0.5*\cellwidth, -2 +
      0.5*\cellwidth) {$a_{\i}'$};
      \else\ifnum\i=7
      \draw[fill=gray!50] (\i*\cellwidth, -2) rectangle
      ++(\cellwidth, \cellwidth);
      \node (a\i') at (\i*\cellwidth + 0.5*\cellwidth, -2 +
      0.5*\cellwidth) {$a_{\i}'$};
      \else
      \draw[fill=white] (\i*\cellwidth, -2) rectangle ++(\cellwidth,
      \cellwidth);
      \node (a\i') at (\i*\cellwidth + 0.5*\cellwidth, -2 +
      0.5*\cellwidth) {$a_{\i}$};
      \fi\fi\fi
    }
\node[left=of a1'] (randomized) {After transfer (new accounts
        $S=\{a_3', a_5',
    a_7'\}$)};
    \node[left=of a1]  (original) {Initial state (old accounts)};

\draw[-Stealth, thick] (a3.south) to (a3'.north);
    \draw[-Stealth, thick] (a5.south) to (a5'.north);
    \draw[-Stealth, thick] (a7.south) to (a7'.north);
\end{tikzpicture}
  \caption{Effect of a successful transfer
    $\tr{S_{\beta'},\pi}$ with $S_{\beta'} = \{a_3', a_5', a_7'\}$ where $S_{\beta} = \{a_3,
    a_5, a_7\}$ is identified by shared index fields ($\mathrm{idx}(a_k)=\mathrm{idx}(a_k')=k$).
    The proof $\pi$ demonstrates ownership and valid rerandomization.
    Shaded accounts are
    the new accounts in masking set $S_{\beta'}$; the old accounts
    $S_{\beta}$ are removed and replaced by $S_{\beta'}$. Accounts outside
  $S_{\beta}$ remain unchanged.}
  \label{fig:randomization}
\end{figure}
The properties and design of the CUAT highlights a key consequence: a
transfer that invalidates accounts in $S_{\beta}$
prevents all concurrent transfers that use any account from
$S_{\beta}$ in their masking
set. This interdependence creates complex synchronization requirements.
Overlapping transfers must coordinate to prevent conflicts, while
disjoint transfers can proceed concurrently.
The atomicity requirement for masking-set updates suggests that CUAT
objects possess non-trivial synchronization power, which we analyze in
the next section.

\subsection{Consensus number}
\label{subsec:cons-num}
This section characterizes the wait-free consensus number of $\lambda$-CUAT.
Protocols may use any finite number of $\lambda$-CUAT instances together with
atomic registers. By Definition~\ref{def:ruat-lambda}, the bijection $\phi$ of
every successful transfer preserves each account's index $\mathrm{idx}(a)$.
Accounts with the same index therefore represent the same logical position
across rerandomizations, and two masking sets conflict precisely when they share
an index. We capture this synchronization structure by a graph on masking sets
(Definition~\ref{def:conflict-graph}).

Two modeling choices distinguish the analysis from earlier work on asset
transfer. Guerraoui \etal~\cite{DBLP:journals/dc/GuerraouiKMPS22} encode the
winning process in account balances. CUAT cannot use this channel because only
an account's owner may query its balance, so its synchronization must instead
run through the masking sets. Moreover, untraceability is a runtime
constraint~\cite{DBLP:conf/wdag/DaianLAG18}: we admit only executions consistent
with the notion in force and take the consensus number over that restricted
family.

The analysis splits along the round structure of $\tau$-invocations, and we begin with the case of practical interest.

Section~\ref{subsubsec:one-round} treats \emph{one-round} protocols (one $\tau$-invocation per process). This models how CUAT is used in payment systems~\cite{DBLP:conf/asiacrypt/FauziMMO19,DBLP:journals/iacr/AlupothaGC24,DBLP:conf/asiacrypt/MadathilS25}: a user submits a single transaction whose masking set is fixed before submission, and validators either commit it or reject it. The one-round upper bound therefore captures the synchronization cost faced inside one such transaction: validators process non-conflicting transfers independently and only need to coordinate when two transfers contend over a shared masking set, and the upper bound sizes the worst-case coordination slice.

Section~\ref{subsubsec:multi-round} lifts the analysis to \emph{multi-round} protocols, where CUAT is used as a synchronization primitive inside a layered agreement protocol whose outer rounds adapt their masking sets to earlier outcomes. This captures CUAT's theoretical synchronization power but is not the typical operational mode of payment systems. 
To start our analysis, we first formalize the conflict graph and its clique number.

\begin{definition}[Conflict graph and clique number]
\label{def:conflict-graph}
  Two masking sets $S$ and $S'$ on a common $\cuat_\lambda$ object
  \emph{conflict} when there exist $a\in S$ and $a'\in S'$ with
  $\mathrm{idx}(a)=\mathrm{idx}(a')$. Let $\mathcal{S}(\Pi)$
  denote the family of masking sets carried by some $\tau$-invocation
  in some execution of $\Pi$ admissible under the untraceability
  notion in force. The \emph{conflict graph} $G$ of $\Pi$ has a vertex
  $v_S$ for each $S\in\mathcal{S}(\Pi)$, and an edge
  $(v_S,v_{S'})\in E$ whenever $S$ and $S'$ conflict. A \emph{clique}
  in $G$ is a set of pairwise-conflicting masking sets, and the
  \emph{clique number} $\omega(G)$ is the size of the largest clique.
\end{definition}

The family $\mathcal{S}(\Pi)$ captures every masking set $\Pi$ may issue, including across executions where processes adapt their choices to the schedule. Any single execution realizes a subgraph of $G$, with clique number at most $\omega(G)$. The upper-bound arguments below inspect the masking set each process is about to invoke at a critical configuration; these masking sets lie in $\mathcal{S}(\Pi)$ by definition, so the bounds hold whether $\Pi$ chooses masking sets statically or dynamically. The matching lower bounds are realized by protocols whose masking-set assignment is fixed in advance.
Alongside the conflict graph, a second invariant tracks how many masking sets contain a given account.
\begin{definition}[Incidence number]
\label{def:incidence}
The \emph{incidence} of an account $a\in\mathcal{A}$ is the number of
vertices $v_S\in V(G)$ such that $a\in S$. The \emph{incidence number} of
$\Pi$ is $r(\Pi)=\max_{a\in\mathcal{A}}\,\mathrm{incidence}(a)$.
\end{definition}
Cliques in $G$ implement consensus. When $d$ masking sets pairwise conflict, at most one of the associated transfers can return $\textsc{true}$: any pair shares an index, and the second transfer to commit finds that index already invalidated and fails. Losing transfers inspect the published CUAT state, locate the successful masking set, and adopt the value its issuer wrote to a single-writer register. Pairing each of $d$ processes with one clique member yields a $d$-process consensus protocol that attempts the transfer, decides on success, and adopts on failure; hence $\omega(G)$ lower-bounds the consensus power of $\lambda$-CUAT.

The matching upper bounds, both at $\omega(G)$ and beyond, rest on a binary-input specialization of the standard critical-state construction~\cite{DBLP:journals/toplas/Herlihy91}.

\begin{lemma}[Different valencies]
\label{lem:different-valency-gen}
Let $\Pi$ be a deterministic wait-free consensus protocol for
$n\ge 2$ processes in the asynchronous shared-memory model. For any
process $p$ and any $v\in\{0,1\}$, the input assignment $I$ with
$I(p)=v$ and $I(s)=1-v$ for $s\ne p$ admits a reachable critical
configuration $C$ such that $C\cdot p$ is $v$-valent and $C\cdot s$
is $(1-v)$-valent for every $s\ne p$.
\end{lemma}

\begin{proof}
By symmetry of $0$ and $1$, take $v=0$; process $p$ then has input
$0$ and every other process has input $1$. Let $\sigma$ be the
schedule in which $p$ takes no steps and every other process runs
to completion. By wait-freedom each $s\ne p$ decides along $\sigma$,
and by validity every such decision is $1$, since $1$ is the only
input present. Hence $\sigma$ ends in a $1$-valent configuration.
The initial configuration is bivalent ($\sigma$ decides $1$ while
running $p$ alone decides $0$), so $\sigma$ has a last bivalent
prefix, and the standard critical-state construction extends that
prefix $p$-silently to a reachable critical configuration $C$.

Fix $s\ne p$ and suppose, toward a contradiction, that $C\cdot s$
is $0$-valent. Every extension of $C\cdot s$ then decides $0$.
Consider the extension that runs every process in $\mathcal{P}\setminus\{p\}$
to completion before $p$ takes a step: wait-freedom guarantees
termination, and validity forces the decision to be $1$, since every
participant before $p$ acts has input $1$. This contradicts
$0$-valence, so $C\cdot s$ is $1$-valent for every $s\ne p$. By
bivalence of $C$, the remaining successor $C\cdot p$ is $0$-valent.
\end{proof}
In each subsequent application, we extract a reachable critical configuration whose successors split processes by valency, then exploit the commutativity of $\lambda$-CUAT transfers on disjoint account indices: two masking sets on opposite valency sides must conflict, since otherwise their transfers commute, and a spectator process, unable to distinguish the two orders, would be forced to decide both values. The next lemma packages this technique; the spectator is why it requires a third process.

\begin{lemma}[Bipartite intersection at a critical configuration]
\label{lem:bipartite-intersection}
Let $\Pi$ be a wait-free binary consensus protocol for $n\ge 3$ processes on $\cuat_\lambda$ objects and atomic registers, and let $C$ be a reachable critical configuration of $\Pi$. There exist a $\cuat_\lambda$ object $O$ and a non-trivial partition $\mathcal{P}_0\sqcup\mathcal{P}_1=\mathcal{P}$ such that every pending step at $C$ is a $\tau$-invocation on $O$, and, writing $S_p$ for the masking set carried by $p$'s pending invocation:
\begin{enumerate}[(i)]
\item $C\cdot p$ is $0$-valent for every $p\in\mathcal{P}_0$ and $C\cdot q$ is $1$-valent for every $q\in\mathcal{P}_1$;
\item $S_p\cap S_q\ne\emptyset$ on $O$ for every $(p,q)\in\mathcal{P}_0\times\mathcal{P}_1$.
\end{enumerate}
\end{lemma}

\begin{proof}
At a critical configuration neither pending step is a register operation~\cite{DBLP:journals/toplas/Herlihy91}: a read leaves the shared state unchanged, writes to distinct registers commute, and of two writes to one register the second overwrites the first, leaving a configuration indistinguishable to every process except the overwritten writer. Each case makes two oppositely valent successors indistinguishable to some process that decides in both. Every pending step at $C$ is therefore a $\tau$-invocation. Define $\mathcal{P}_0=\{p\in\mathcal{P}:C\cdot p\text{ is }0\text{-valent}\}$ and $\mathcal{P}_1=\mathcal{P}\setminus\mathcal{P}_0$; both are non-empty by bivalency of $C$, and property~(i) follows.

Fix $p\in\mathcal{P}_0$ and $q\in\mathcal{P}_1$, and suppose toward a contradiction that $\tau(S_p)$ and $\tau(S_q)$ are on distinct $\lambda$-CUAT objects, or on a common object with $S_p\cap S_q=\emptyset$. Each transfer's outcome is then unaffected by the other: the two masking sets are disjoint, and by account-representation uniqueness neither transfer's new accounts can appear in the other's masking set. The object states after the pair of steps therefore agree in either order, and neither step writes a register, so the shared state at $C\cdot p\cdot q$ equals that at $C\cdot q\cdot p$. The two configurations are not indistinguishable to $p$ and $q$ themselves, whose transfer responses carry the current account set and hence reveal the order; but they are indistinguishable to every other process. Since $n\ge 3$, choose a spectator $s\notin\{p,q\}$ and run it alone from each configuration. By wait-freedom $s$ decides in both runs, and indistinguishability makes the two decisions equal; yet every extension of $C\cdot p\cdot q$ decides $0$, as $C\cdot p$ is $0$-valent, and every extension of $C\cdot q\cdot p$ decides $1$, as $C\cdot q$ is $1$-valent. This contradiction shows that $p$ and $q$ invoke $\tau$ on a common $\lambda$-CUAT object with $S_p\cap S_q\ne\emptyset$ on that object. Iterating over $(p,q)\in\mathcal{P}_0\times\mathcal{P}_1$, all cross-valency pairs share a single object $O$; transitivity through such pairs forces every process in $\mathcal{P}$ to invoke $\tau$ on $O$, and property~(ii) holds with masking sets on $O$.
\end{proof}

In particular, applying the lemma to the critical configuration produced by Lemma~\ref{lem:different-valency-gen} with $(p,v)$ gives the \emph{star} specialization, with $\mathcal{P}_v=\{p\}$ and $\mathcal{P}_{1-v}=\mathcal{P}\setminus\{p\}$: $p$'s pending masking set conflicts on a common object with every other process's pending masking set. Iterating over $p$ is what the one-round upper bound relies on.

\subsubsection{One-round protocols}
\label{subsubsec:one-round}

Throughout this subsection, $\Pi$ denotes a wait-free one-round consensus protocol on $\lambda$-CUAT objects and atomic registers: each process makes a single $\tau$-invocation and identifies a contested winner from the published CUAT state. Before its $\tau$-step a process has access only to the CUAT's initial state and its own identity, so its masking set is determined by the protocol from these two pieces of data alone. A process may also read registers, but atomic registers have consensus number $1$~\cite{DBLP:journals/toplas/Herlihy91} and add no synchronization power on top of $\lambda$-CUAT. The analysis pairs a lower bound from constructive cliques in $G$ with an upper bound from iterating Lemma~\ref{lem:bipartite-intersection} over each process.

Any clique of $G$ admits a wait-free consensus protocol of matching size, by the standard attempt-and-adopt construction.

\begin{lemma}[One-round lower bound from cliques]
\label{lem:single-tau-lower}
Let $G$ be a conflict graph admissible under the untraceability notion in force, and let $K=\{S_1,\dots,S_n\}\subseteq V(G)$ be a clique of size $n$. Then there is a wait-free one-round $n$-consensus protocol $\Pi$ over a $\cuat_\lambda$ object in which process $p_i$ uses masking set $S_i$. Consequently $\consnum{\cuat_\lambda}\ge\omega(G)$.
\end{lemma}

\begin{proof}
\emph{Construction.} Assign $S_i$ as $p_i$'s masking set. Each $p_i$ writes its input to a single-writer register $r_i$ (proposal) and invokes $\tau(S_i)$ on the $\lambda$-CUAT object. On success ($\tau$ returns $\textsc{true}$), $p_i$ decides its own input. On failure, $p_i$ reads the published CUAT state to identify which masking set $S_j\in K$ has been committed, reads $r_j$, and decides that value. Figure~\ref{alg:consensus-from-cuat} gives the full pseudocode; the function $\op{weak}$ selects masking sets under weak untraceability and $\op{strong}$ uses the projective-plane construction under strong untraceability (Section~\ref{subsubsec:strong-untr}).

\emph{Wait-freedom.} The protocol contains no loops: $\tau$ is invoked once, and identification of the winner is a single inspection of the published state. Each process completes in a bounded number of steps regardless of the schedule.

\emph{Validity.} The decided value is either $p_i$'s own input (on success) or $r_{j^*}$ for some $j^*$, which is $p_{j^*}$'s input.

\emph{Agreement.} Linearize the $n$ $\tau$-invocations on $K$ and let $S_{j^*}\in K$ be the first to return $\textsc{true}$. For every other $S_i\in K$, $S_{i}\cap S_{j^*}\ne\emptyset$ by the clique property of $K$; the account shared between $S_i$ and $S_{j^*}$ is invalidated when $S_{j^*}$ commits, so $\tau(S_i)$ fails. Process $p_i$ then identifies $S_{j^*}$ as the winner in the published state and decides $r_{j^*}$. Process $p_{j^*}$ also decides $r_{j^*}$ on success. Every process therefore decides the same value $r_{j^*}$.
\end{proof}

\begin{figure}[ht!]
  \begin{tcolorbox}[protocol={Consensus from $\lambda$-CUAT}]
    \step{\textbf{Shared state:}}
    \begin{codelines}
    \item $R[j]$ for $j\in\{1,\ldots,d\}$, atomic registers (proposals)
    \item $A[j]$ for $j\in\{1,\ldots,d\}$, atomic registers (new account sets)
    \item A $\lambda$-CUAT object $\cuat$ with initial state $(\mathcal{A}_0,\beta_0)$
      where $\mathcal{A}_0=\{a_1,\dots,a_{d+1}\}$ (with $d\ge\lambda$),
      $\mu(a_1)=\{p_1\}$, $\mu(a_{i+1})=\{p_i\}$ for $i\in\{1,\dots,d\}$, and
      $\beta_0(a)=0$ for all $a\in\mathcal{A}_0$; process $p_1$ owns
      $a_1$ which doubles as the hub account
    \item An untraceability level $\op{weak}\in\set{0,1}$
    \end{codelines}
    \step{$\mathbf{propose}(v)$ at process $p_i$ where $\mu(a_{i+1})=\{p_i\}$:}
    \begin{codelines}[resume]
    \item $R[i].\op{write}(v)$ \textit{// Write proposal}
    \item $S_{\text{old}} \gets \op{masking}(i, \lambda, \op{weak})$ \textit{// Old masking set}
    \item $(S_{\text{new}}, \pi) \gets \op{Compute}_{p_i}(S_{\text{old}}, a_{i+1}, a'_{i+1}, 0)$ \textit{// Generate new accounts and proof}
    \item $A[i].\op{write}(S_{\text{new}})$ \textit{// Write new account set}
    \item $(\mathcal{A}, b) \gets \cuat.\tr{S_{\text{new}}, \pi}$
    \item \textbf{if } $b = \true$ \textbf{ then}
    \item \quad \textbf{return } $v$ \textit{// Won: decide own value}
    \item \textbf{else} \textit{// Lost: identify winner by reading account sets}
    \item \quad $S_{\text{winner}} \gets \mathcal{A} \setminus \mathcal{A}_0$ \textit{// New accounts in $\mathcal{A}$}
    \item \quad \textbf{for } $j \in \{1,\ldots,d\}$ \textbf{ do}
    \item \quad \quad \textbf{if } $A[j].\op{read}() = S_{\text{winner}}$ \textbf{ then} \textit{// Process $j$ won}
    \item \quad \quad \quad \textbf{return } $R[j].\op{read}()$
      \textit{// Adopt winner's value}
  \end{codelines}
    \stepmin{\textit{function:} $\op{masking}(i, \lambda, \op{weak})$ :}
    \begin{codelines}[resume]
    \item \textbf{if } $\op{weak}=1$ \textbf{then}
    \item \quad $\op{weak}(i,\lambda)$
    \item \textbf{else} \textit{// Strong untraceability}
    \item \quad $\op{strong}(i,\lambda)$
    \end{codelines}
    \stepmin{\textit{function:} $\op{weak}(i, \lambda)$ :}
    \begin{codelines}[resume]
    \item $S' \gets \{1\}$ \textit{// Hub account $a_1$, owned by $p_1$}
    \item $j \gets i+1$ \textit{// Start with $p_i$'s account $a_{i+1}$}
    \item \textbf{repeat } $\lambda-1$ \textbf{ times}
    \item \quad $S' \gets S' \cup \{j\}$
    \item \quad $j \gets (j \bmod (d+1)) + 1$
    \item \quad \textbf{if } $j = 1$ \textbf{ then } $j \gets 2$ \textit{// Skip hub}
    \item \textbf{return} $S'$
    \end{codelines}
    \stepmin{\textit{function:} $\op{strong}(i, \lambda)$ :
    (see~\cite{IHRINGER2019142})}
\end{tcolorbox}
  \vspace{-0.5cm}
  \caption{Wait-free implementation of $d$-consensus using a
  $\lambda$-CUAT object.}
  \label{alg:consensus-from-cuat}
\end{figure}
Conversely, Lemma~\ref{lem:bipartite-intersection} turns wait-freedom into a clique witness in $G$.
\begin{lemma}[One-round upper bound from pairwise intersection]
\label{lem:single-tau-upper}
Let $\Pi$ be a wait-free one-round consensus protocol on $\cuat_\lambda$ objects and atomic registers, with $G$ its conflict graph. If $cn(\Pi)\ge 3$, then $cn(\Pi)\le\omega(G)$.
\end{lemma}
\begin{proof}
Suppose, for contradiction, that $\Pi$ solves $n$-consensus for $n=\omega(G)+1\ge 3$, so the largest clique of $G$ has size $n-1$. Write $S_i$ for the masking set used by $p_i$ for $i=1,\dots,n$.

\emph{A non-conflicting pair exists.} If every pair $\{S_i,S_j\}$ conflicted, the vertices $v_{S_1},\dots,v_{S_n}$ would form an $n$-clique in $G$ (Definition~\ref{def:conflict-graph}), contradicting $\omega(G)=n-1$. Hence there exist processes $p\ne q$ whose masking sets $S_p,S_q$ do not conflict: either they sit on distinct $\lambda$-CUAT objects or they share no index on a common object.

\emph{Critical configuration with different valencies.} Set $q$'s input to $0$ and every other input to $1$, and apply Lemma~\ref{lem:different-valency-gen}: there is a reachable critical configuration $C$ at which $C\cdot q$ is $0$-valent, $C\cdot s$ is $1$-valent for every $s\ne q$, and every process's pending step is its first $\tau$-invocation. At a critical configuration neither pending step is a register operation~\cite{DBLP:journals/toplas/Herlihy91}: a read leaves the shared state unchanged, writes to distinct registers commute, and of two writes to one register the second overwrites the first, leaving a configuration indistinguishable to every process except the overwritten writer. Each case makes two oppositely valent successors indistinguishable to some process that decides in both. The pending step of $p$ at $C$ is therefore a $\tau$-invocation carrying $S_p$, and likewise for $q$ with $S_q$. In particular, $C\cdot p$ is $1$-valent and $C\cdot q$ is $0$-valent.

\emph{Equal shared state.} Since $S_p$ and $S_q$ do not conflict, $\tau(S_p)$ and $\tau(S_q)$ act on disjoint state: the outcome of each is unaffected by the other, and the $\lambda$-CUAT object(s) reach the same state after the pair in either order. Neither step modifies a register, so the shared state at $C\cdot p\cdot q$ equals that at $C\cdot q\cdot p$. The two orders are visible to $p$ and $q$ themselves, whose transfer responses carry the current account set, but to no other process.

\emph{A spectator derives the contradiction.} Since $n\ge 3$, choose $s\notin\{p,q\}$ and run it alone from each of the two configurations. By wait-freedom $s$ decides in both, and because its local state and the shared state agree in the two configurations, it decides the same value $v$ in both. But every extension of $C\cdot p\cdot q$ decides $1$ because $C\cdot p$ is $1$-valent, forcing $v=1$, while every extension of $C\cdot q\cdot p$ decides $0$ because $C\cdot q$ is $0$-valent, forcing $v=0$, a contradiction.

The contradiction rules out $cn(\Pi)>\omega(G)$, so $cn(\Pi)\le\omega(G)$.
\end{proof}
Combining the two bounds gives equality.
\begin{theorem}[One-round consensus number]
\label{thm:single-phase-equals-omega}
If $\Pi$ is a one-round protocol on $\cuat_\lambda$ with $cn(\Pi)\ge 3$, then $cn(\Pi)=\omega(G)$. Consequently, the largest $n$ achievable by a one-round $n$-consensus protocol on $\cuat_\lambda$ equals the supremum of $\omega(G)$ over conflict graphs admissible under the untraceability notion in force. 
\end{theorem}
\begin{proof}
Combine Lemma~\ref{lem:single-tau-lower} and Lemma~\ref{lem:single-tau-upper}.
\end{proof}
Theorem~\ref{thm:single-phase-equals-omega} matches the one-round upper bound to $\omega(G)$, and the bound is the limit of the one-round regime: a committed $\tau$ permanently invalidates its masking set, and the surviving masking sets that conflicted with it can no longer fire, so the conflict graph is exhausted after one round of contention. Two follow-up questions remain. Section~\ref{subsubsec:multi-round} relaxes the one-round restriction and asks how far multi-round protocols can push the consensus number past $\omega(G)$. Sections~\ref{subsubsec:weak-untr} and~\ref{subsubsec:strong-untr} take the orthogonal question of how large $\omega(G)$ can be made under each untraceability notion: weak untraceability leaves $\omega(G)$ unbounded, while strong untraceability caps it at $\lambda^2-\lambda+1$.

\subsubsection{Multi-round protocols}
\label{subsubsec:multi-round}
A multi-round protocol lets each process invoke $\tau$ more than once and adapt later invocations to earlier outcomes. 
The structural gain over $\omega(G)$ is not larger cliques but a weaker requirement on the conflict graph: pairwise intersection at the critical configuration can be relaxed once the pending processes split by valency, so the next $\tau$-round need only enforce conflict across the partition. Lemma~\ref{lem:bipartite-intersection} already extracts such a partition $\mathcal{P}=\mathcal{P}_0\sqcup\mathcal{P}_1$ with a complete bipartite intersection between the two sides. At incidence $2$ a $\lambda$-block biclique supports $2\lambda$ vertices, exceeding the one-round upper bound.
\begin{theorem}[Upper bound]
\label{thm:multi-phase-upper}
A wait-free binary consensus protocol $\Pi$ over $\cuat_\lambda$ objects and atomic registers can solve consensus among at most $\max\bigl(2,\,2\lambda(r(\Pi)-1)\bigr)$ processes.
\end{theorem}

\begin{proof}
Suppose $\Pi$ solves $n$-consensus, so $|\mathcal{P}|=n$; for $n\le 2$ the bound is immediate, so assume $n\ge 3$. Take $C$, $\mathcal{P}_0$, $\mathcal{P}_1$ as in Lemma~\ref{lem:bipartite-intersection}; the partition covers $\mathcal{P}$, and cross-valency pending masking sets conflict on a common $\cuat_\lambda$ object.

Two processes \emph{across} the valency partition use distinct pending masking sets. Were $p\in\mathcal{P}_0$ and $q\in\mathcal{P}_1$ to share $S_p=S_q=S$, the winner-identification rule of Section~\ref{subsubsec:one-round} would match the single committed $S$ in the published CUAT state to both $p$ and $q$, and the two valency classes (which decide $0$ and $1$ respectively in any extension, by property~(i) of Lemma~\ref{lem:bipartite-intersection}) would have no consistent register to adopt, breaking agreement. Same-side sharing is benign: two processes in the same valency class agree on the decision value of any extension, so adopting either's register is consistent. Hence $\{v_{S_p}:p\in\mathcal{P}_0\}\cap\{v_{S_q}:q\in\mathcal{P}_1\}=\emptyset$ in $V(G)$, and property~(ii) of Lemma~\ref{lem:bipartite-intersection} yields a complete bipartite subgraph $K_{|\mathcal{P}_0|,|\mathcal{P}_1|}\subseteq G$ between these two vertex sets.

Fix $p\in\mathcal{P}_0$. Each account $a\in S_p$ has incidence at most $r(\Pi)$ in $G$, hence lies in at most $r(\Pi)-1$ vertices distinct from $v_{S_p}$. Summing over the $\lambda$ accounts of $S_p$, at most $\lambda(r(\Pi)-1)$ distinct vertices in $V(G)$ are adjacent to $v_{S_p}$ via a shared account. Since every $v_{S_q}$ with $q\in\mathcal{P}_1$ is adjacent to $v_{S_p}$, we get $|\mathcal{P}_1|\le\lambda(r(\Pi)-1)$. By symmetry $|\mathcal{P}_0|\le\lambda(r(\Pi)-1)$, so $n=|\mathcal{P}_0|+|\mathcal{P}_1|\le 2\lambda(r(\Pi)-1)$.
\end{proof}

In the following two sections, we discuss how the incidence number $r(\Pi)$ relates to privacy and, through it, to the consensus number.

\subsubsection{Weak untraceability}
\label{subsubsec:weak-untr}

Under weak untraceability the incidence number $r(\Pi)$ is unconstrained: an
account may appear in every masking set without violating per-transaction anonymity.
The hub construction of Figure~\ref{alg:consensus-from-cuat} (function $\op{weak}$)
places the shared hub account $a_1$ (owned by $p_1$) in all $d$ masking sets,
making the conflict graph a $d$-clique, and satisfies weak untraceability for any
$d$. By Theorem~\ref{thm:single-phase-equals-omega} this yields $cn(\Pi)=d$, so
the consensus number is unbounded.

\begin{theorem}[Consensus number under weak untraceability]
  \label{thm:weak-untraceability}
  For every $d\ge 2$, $\cuat_\lambda^-$ supports a wait-free one-round $d$-consensus protocol. Hence $\consnum{\cuat_\lambda^-}=\infty$.
\end{theorem}

\begin{proof}
We instantiate Lemma~\ref{lem:single-tau-lower} on the hub family and verify that the family satisfies weak untraceability.

\emph{Construction.} Use the hub function of Figure~\ref{alg:consensus-from-cuat}: every masking set $S_i$ contains the hub $a_1$ (owned by $p_1$) together with $\lambda-1$ further accounts chosen cyclically. Since all $d$ sets share $a_1$, they pairwise intersect, so $K=\{S_1,\dots,S_d\}$ is a $d$-clique in the conflict graph $G$. Lemma~\ref{lem:single-tau-lower} then implements wait-free $d$-consensus on this clique.

\emph{Weak untraceability.} The cryptographic proof $\pi$ of each $\tau$-invocation is zero-knowledge: an observer of $\tau$ learns the masking set $S_{\mathrm{old}}$ and the published state, but obtains no advantage over the uniform distribution on $S_{\mathrm{old}}$ for the sender's identity. Hence $P(a_i\mid\tau)=1/|S_{\mathrm{old}}|=1/\lambda$ for every $a_i\in S_{\mathrm{old}}$, regardless of $d$ or of the hub-account structure.

Therefore $\Pi$ achieves wait-free $d$-consensus and preserves weak untraceability for every $d\ge 2$.
\end{proof}

The hub construction leaves $a_1$ in every masking set, so a history-aware adversary identifies it as a near-certain decoy; this statistical leakage motivates strong untraceability.

\subsubsection{Strong untraceability}
\label{subsubsec:strong-untr}

Strong untraceability requires anonymity against adversaries that observe the
full transaction history. The masking sets of a consensus protocol pairwise
intersect, so Corollary~\ref{cor:uniform-incidence} applies and they must have
uniform incidence. We first combine this condition with
the clique characterization of one-round consensus, and then with the
bipartite-intersection characterization of unrestricted protocols.

The one-round bound is an extremal question about set systems, and a projective
plane is its extremal case. A one-round protocol needs its $n$
masking sets to pairwise intersect, while uniform incidence forbids any account
from carrying more of those intersections than the others. Suppose every account
lies in $m$ of the sets, and write $v$ for the number of accounts. Counting
incidences two ways gives $vm=n\lambda$, and an account of incidence $m$ accounts
for at most $\binom{m}{2}$ of the $\binom{n}{2}$ pairs that must meet, so
$\binom n2\le v\binom m2$. Substituting $v=n\lambda/m$ leaves $n\le\lambda(m-1)+1$,
and $m\le\lambda$, so $n\le\lambda^2-\lambda+1$. Equality forces both counts to
be tight at once: every account lies in exactly $\lambda$ sets, and every two
sets meet in exactly one account, since a second shared account would spend
incidence on a pair already covered. Those two conditions are the axioms of a
projective plane of order $\lambda-1$, read with accounts as points and masking
sets as lines. A family of size $\lambda^2-\lambda+1$ is therefore a projective plane, and the
bound is attained exactly when one exists.

\begin{theorem}[One-round upper bound under strong untraceability]
  \label{thm:one-round-strong-upper}
  A wait-free one-round consensus protocol on $\cuat_\lambda^+$ can solve consensus among at most $\lambda^2-\lambda+1$ processes.
\end{theorem}

\begin{proof}
  Let $\Pi$ be a $d$-consensus protocol of the class of
  Section~\ref{subsubsec:one-round}, with $d\ge 3$, and let
  $F=\{S_1,\dots,S_d\}$ be its masking sets. By
  Lemma~\ref{lem:single-tau-upper} they lie on a common object and form a
  $d$-clique in $G$, so they pairwise intersect; each has size $\lambda$ by
  definition; and by Corollary~\ref{cor:uniform-incidence} strong
  untraceability forces every account of $\bigcup_iS_i$ to lie in a common
  number $m$ of them. F\"{u}redi~\cite{furedi1981maximum} and Ihringer and
  Kupavskii~\cite{IHRINGER2019142} bound such symmetric pairwise-intersecting
  $\lambda$-uniform families by $|F|\le\lambda^2-\lambda+1$. Hence
  $d\le\lambda^2-\lambda+1$.
\end{proof}

The upper bound is attained by a projective plane whenever $\lambda-1$ is a
prime power.

\begin{theorem}[Projective-plane construction]
\label{thm:single-tau-strong-tight}
If $\lambda-1$ is a prime power, $\cuat_\lambda^+$ supports a wait-free one-round $(\lambda^2-\lambda+1)$-consensus protocol, instantiated by attempt-and-adopt on the line set of $\mathrm{PG}(2,\lambda-1)$.
\end{theorem}

\begin{proof}
Let $q=\lambda-1$. The projective plane $\mathrm{PG}(2,q)$ has $\lambda^2-\lambda+1$ points and the same number of lines; each line contains $\lambda=q+1$ points, every two lines meet in exactly one point, and each point lies on exactly $\lambda$ lines~\cite{IHRINGER2019142}. Interpreting points as accounts and lines as masking sets yields a $\lambda$-uniform pairwise-intersecting family of size $\lambda^2-\lambda+1$ over $\lambda^2-\lambda+1$ accounts, of uniform incidence $\lambda$, hence admissible under strong untraceability by Theorem~\ref{thm:untraceability-uniform}.

To run the attempt-and-adopt protocol of Lemma~\ref{lem:single-tau-lower} (Figure~\ref{alg:consensus-from-cuat}) on this family, each of the $\lambda^2-\lambda+1$ processes must be assigned a distinct masking set containing one of that process's accounts. This is a perfect-matching problem in the bipartite graph between accounts and masking sets, with an edge for each containment. Each account appears in exactly $\lambda$ masking sets and each masking set contains exactly $\lambda$ accounts, so the graph is $\lambda$-regular; König's theorem~\cite{west2001introduction} yields a perfect matching. Lemma~\ref{lem:single-tau-lower} then implements $(\lambda^2-\lambda+1)$-consensus on this family.
\end{proof}

For other $\lambda$ no plane exists, and the one we have cannot be cut down to
size: if a sub-family of $n$ lines still covering all $\lambda^2-\lambda+1$
points had uniform incidence $m$, then $(\lambda^2-\lambda+1)m=n\lambda$, and
$\lambda^2-\lambda+1\equiv 1\pmod\lambda$ makes the two factors coprime, forcing
$\lambda\mid m$, hence $m=\lambda$ and $n=\lambda^2-\lambda+1$, the whole line
set. Uniform incidence is not hereditary, so a smaller family has to be built
directly. We use a cyclic construction. In a family of translates
$\{D+t\}_{t\in\mathbb{Z}_n}$, an account $x$ lies in $D+t$ exactly when
$t\in x-D$, so its incidence is $|D|$ for any choice of $D$. Pairwise
intersection then reduces to a condition on the differences of $D$, which can
be satisfied at size $\lfloor\lambda^2/2\rfloor$ for every $\lambda$.

\begin{theorem}[Cyclic construction]
\label{thm:cyclic-one-round}
For every $\lambda\ge 3$, $\cuat_\lambda^+$ supports a wait-free one-round
$\lfloor\lambda^2/2\rfloor$-consensus protocol, instantiated by attempt-and-adopt
on the translates of a cyclic difference cover.
\end{theorem}

\begin{proof}
Put $g=\lfloor\lambda/2\rfloor$, $h=\lceil\lambda/2\rceil$ and
$n=\lfloor\lambda^2/2\rfloor$, so that $g+h=\lambda$ and, checking both parities
($\lambda=2k$ gives $gh=k^2$ and $n=2k^2$; $\lambda=2k+1$ gives $gh=k(k+1)$ and
$n=2k(k+1)$),
\begin{equation}
\label{eq:gh-half-n}
  gh=\lfloor\lambda/2\rfloor\lceil\lambda/2\rceil=\frac n2 .
\end{equation}
Identify both the accounts and the processes with $\mathbb{Z}_n$, let
$\mu(x)=\{x\}$, and set
\[
  D=\{0,1,\dots,g-1\}\cup\{g,2g,\dots,hg\}\subseteq\mathbb{Z}_n ,
\]
a run of $g$ consecutive residues followed by the first $h$ multiples of $g$.
The masking-set family is the set of translates $S_t=D+t$, $t\in\mathbb{Z}_n$,
with $S_t$ assigned to process $t$.

\emph{Size and incidence.} The run and the multiples are disjoint, as $jg\ge g>g-1$
for $j\ge 1$, and $g,2g,\dots,hg$ are distinct as $hg=n/2<n$; hence
$|D|=g+h=\lambda$ and every translate has size $\lambda$. For $x\in\mathbb{Z}_n$
we have $x\in D+t$ iff $t\in x-D$, so $x$ lies in exactly $|D|=\lambda$
translates. The account set is $\bigcup_t S_t=\mathbb{Z}_n$ and every account has
incidence $\lambda$, so the family is admissible under strong untraceability by
Theorem~\ref{thm:untraceability-uniform}.

\emph{Pairwise intersection.} As $(D+t)\cap(D+t')\ne\emptyset$ iff
$t'-t\in D-D$, the translates pairwise intersect iff $D-D=\mathbb{Z}_n$. The run
contributes $[-(g-1),g-1]$, and for each $j\in\{1,\dots,h\}$ the multiple $jg$
contributes $jg-\{0,\dots,g-1\}=[(j-1)g+1,\,jg]$; these $h$ blocks tile $[1,hg]$
without gaps. Hence $D-D\supseteq[-hg,hg]$, which modulo $n$ is
$[0,n/2]\cup[n/2,n-1]=\mathbb{Z}_n$ because $hg=n/2$ by \eqref{eq:gh-half-n}. So
the $n$ vertices $v_{S_t}$ form an $n$-clique of $G$.

\emph{Distinct masking sets.} The translates are pairwise distinct iff the
stabilizer $H=\langle d\rangle$ of $D$ is trivial. As $0\in D$ and $D$ is a union
of $H$-cosets, $H\subseteq D\subseteq[0,hg]=[0,n/2]$; from $n-d\in H$ we get
$n-d\le n/2$, so $d\ge n/2$. If $d=n/2$ then $D$ is a union of pairs
$\{x,x+n/2\}$, which lie in $[0,n/2]$ only for $x=0$, forcing $D=\{0,n/2\}$ and
$\lambda=2$, excluded. Hence $d=n$ and $H$ is trivial.

\emph{Assignment.} Since $0\in D$, process $t$ owns the account $t\in S_t$, so
$t\mapsto S_t$ gives each process a distinct masking set containing one of its
own accounts, and Lemma~\ref{lem:single-tau-lower} applied to this $n$-clique
implements wait-free one-round $n$-consensus.
\end{proof}

Combining the upper bound with the two constructions yields the following.

\begin{corollary}[One-round consensus number under strong untraceability]
\label{cor:one-round-sandwich}
For every $\lambda\ge 3$, the largest $n$ for which $\cuat_\lambda^+$ supports a
wait-free one-round $n$-consensus protocol satisfies
$\lfloor\lambda^2/2\rfloor\le n\le\lambda^2-\lambda+1$, with equality on the
right when $\lambda-1$ is a prime power.
\end{corollary}

\begin{proof}
The upper bound is Theorem~\ref{thm:one-round-strong-upper}, the lower
bound is Theorem~\ref{thm:cyclic-one-round}, and the prime-power case is
Theorem~\ref{thm:single-tau-strong-tight}.
\end{proof}

We now turn to unrestricted protocols, where the demand on the masking sets is
weaker. A one-round protocol needs every pair of its sets to meet; by
Lemma~\ref{lem:bipartite-intersection} an unrestricted protocol needs this only
across the two valency classes, leaving the pairs inside a class free. Fewer pairs must be
covered at the same incidence, so the bound is larger. It follows by counting
how much of the complete bipartite intersection a single account can cover under
uniform incidence.

\begin{theorem}[Upper bound under strong untraceability]
\label{thm:multi-phase-strong-refined}
For every $\lambda\ge 2$, a wait-free consensus protocol on $\cuat_\lambda^+$ can
solve consensus among at most $4\lfloor\lambda/2\rfloor\lceil\lambda/2\rceil$
processes, that is, at most $\lambda^2$ if $\lambda$ is even and at most
$\lambda^2-1$ if $\lambda$ is odd.
\end{theorem}

\begin{proof}
Suppose $\Pi$ solves $n$-consensus; if $n\le 2$ the bound holds trivially,
since $4\lfloor\lambda/2\rfloor\lceil\lambda/2\rceil\ge 4$, so assume
$n\ge 3$. Take $C$, $\mathcal{P}_0$, $\mathcal{P}_1$, and $O$ as in
Lemma~\ref{lem:bipartite-intersection}, and write
$s=|\mathcal{P}_0|+|\mathcal{P}_1|$. Uniformity first bounds the incidence of
an account by $\lambda$: if $d$ masking sets of size $\lambda$ cover $v$
accounts with common incidence $m$, then $vm=d\lambda$ by double counting,
while Fisher's inequality gives $d\le v$, and hence $m\le\lambda$.

For every account $a$ on $O$, let $i_a$ and $j_a$ count its incidences among
the pending masking sets of $\mathcal{P}_0$ and $\mathcal{P}_1$, respectively.
Every cross-class pair must share an account, so
$|\mathcal{P}_0|\cdot|\mathcal{P}_1|\le\sum_a i_a j_a$. Moreover,
$\sum_a(i_a+j_a)=s\lambda$ and $i_a+j_a\le\lambda$. Put
\[
  \kappa:=\max_{\substack{i+j\le\lambda\\i,j\ge 0}}ij=\lfloor\lambda/2\rfloor\lceil\lambda/2\rceil=\begin{cases}\lambda^2/4 & \lambda\text{ even,}\\ (\lambda^2-1)/4 & \lambda\text{ odd,}\end{cases}
\]
attained when $i$ and $j$ are balanced and $i+j=\lambda$. For every
$i+j\le\lambda$, we have $ij\le(\kappa/\lambda)(i+j)$, and therefore
$\sum_a i_a j_a\le(\kappa/\lambda)\sum_a(i_a+j_a)=s\kappa$. The balanced
input assignment forces $|\mathcal{P}_0|=|\mathcal{P}_1|=s/2$ at the critical
configuration, so
$
  \tfrac{s^2}{4}=|\mathcal{P}_0|\cdot|\mathcal{P}_1|\le\sum_a i_a j_a\le s\,\kappa,
$
giving $s\le 4\kappa=4\lfloor\lambda/2\rfloor\lceil\lambda/2\rceil$, i.e., $\lambda^2$ for even $\lambda$ and $\lambda^2-1$ for odd.
\end{proof}

To attain the bound, split the processes into two groups of
$n=\lfloor\lambda^2/2\rfloor$. Each group first uses
Theorem~\ref{thm:cyclic-one-round} to agree on one value. A second round then
selects between the two group values. The masking sets of this round need only
intersect across the groups, which motivates the following definition.

\begin{definition}[Cross-cover family]
\label{def:cross-cover}
Let $\mathcal{P}_0$ and $\mathcal{P}_1$ be disjoint sets of $n$ processes each. A
\emph{cross-cover family of order $n$} assigns to every $r\in\mathcal{P}_0\cup\mathcal{P}_1$
a masking set $S^{(2)}_r$ on one common $\cuat_\lambda$ object $O$ such that
\begin{enumerate}[(i)]
  \item $|S^{(2)}_r|=\lambda$ for every $r$;
  \item every account of $O$ lies in exactly $\lambda$ of the sets $S^{(2)}_r$; and
  \item $S^{(2)}_p\cap S^{(2)}_q\ne\emptyset$ for every
    $(p,q)\in\mathcal{P}_0\times\mathcal{P}_1$.
\end{enumerate}
\end{definition}

Condition~(iii) is exactly the bipartite intersection that
Lemma~\ref{lem:bipartite-intersection} extracts from any wait-free protocol at a
critical configuration; conditions~(i) and~(ii) are imposed by strong
untraceability. No intersection is required \emph{within} a class, so a
cross-cover family may exceed the one-round bound.

\begin{lemma}[Composition]
\label{lem:cross-cover-compose}
Suppose $\cuat_\lambda^+$ supports a wait-free one-round $n$-consensus protocol,
and that a cross-cover family of order $n$ exists. Then $\cuat_\lambda^+$
supports a wait-free $2n$-consensus protocol, so
$\consnum{\cuat_\lambda^+}\ge 2n$.
\end{lemma}

\begin{proof}
Let $\mathcal{P}_0,\mathcal{P}_1$ be the two classes of the cross-cover family,
of $n$ processes each, and give each class its own $\cuat_\lambda$ object on
which to run the one-round protocol.

\emph{Round 1.} Each class runs the one-round $n$-consensus protocol among its
own $n$ processes, on its own object. Every process of class $v$ thus decides a
common \emph{class value} $w_v$, which is the input of some process of that
class, and the round is admissible under strong untraceability by hypothesis.

\emph{Round 2.} Each process $r$ writes $w_v$ to a single-writer register and
invokes $\tau$ on its masking set $S^{(2)}_r$ on the shared object $O$. By
condition~(iii) of Definition~\ref{def:cross-cover}, the masking sets of any
cross-class pair meet, so once some transfer commits, every transfer of the
\emph{other} class finds a shared account invalidated and fails. All committing
transfers therefore belong to one class, the \emph{winning} class, and at least
one commits, namely the first to linearize. A process that commits decides its
own class value; one that fails inspects the published state of $O$, identifies a
committed masking set, reads off the class of its owner, and adopts that class's
value from the corresponding register.

Every process thus decides the winning class's value $w_v$, which round 1 made
common to that class and which is some process's input, so agreement and validity
hold; both rounds are loop-free and wait-free, so the composition is. Round 2 is
admissible under strong untraceability by conditions~(i) and~(ii) of
Definition~\ref{def:cross-cover} with Theorem~\ref{thm:untraceability-uniform}.
\end{proof}

Lemma~\ref{lem:cross-cover-compose} requires its two ingredients at the same
order, so the cross-cover family must match the order that
Theorem~\ref{thm:cyclic-one-round} supplies for the first round. It remains to
construct one at order $\lfloor\lambda^2/2\rfloor$, which the two parities
require us to do separately.

\begin{lemma}[Cross-cover families exist at order $\lfloor\lambda^2/2\rfloor$]
\label{lem:cross-cover-exists}
For every $\lambda\ge 3$ there is a cross-cover family of order
$n=\lfloor\lambda^2/2\rfloor$.
\end{lemma}

\begin{proof}
\emph{Even $\lambda$: a grid.} Here $n=\lambda^2/2$. Partition $\mathcal{P}_0$
into $\lambda$ blocks $P_1,\dots,P_\lambda$ of size $\lambda/2$, and likewise
$\mathcal{P}_1$ into $Q_1,\dots,Q_\lambda$; this is possible since
$|\mathcal{P}_v|=n=\lambda\cdot(\lambda/2)$. On a new object instance take accounts
$\{a_{ij}\}_{i,j\in[\lambda]}$ and set
\[
  S^{(2)}_p=\{a_{ij}:j\in[\lambda]\}\ \ (p\in P_i),
  \qquad
  S^{(2)}_q=\{a_{ij}:i\in[\lambda]\}\ \ (q\in Q_j).
\]
Each masking set is a row or a column of the grid, hence of size $\lambda$,
giving~(i). The account $a_{ij}$ lies in $S^{(2)}_p$ exactly for the $\lambda/2$
processes $p\in P_i$ and in $S^{(2)}_q$ exactly for the $\lambda/2$ processes
$q\in Q_j$, so its incidence is $\lambda/2+\lambda/2=\lambda$, giving~(ii). For
a cross pair $(p,q)\in P_i\times Q_j$, row $i$ and column $j$ meet exactly in
$a_{ij}$, giving~(iii). Figure~\ref{fig:cross-cover-grid} illustrates the
construction for $\lambda=4$.

\begin{figure}[ht]
  \centering
  \begin{tikzpicture}[x=1.15cm,y=1.05cm, every node/.style={font=\small}]
    \foreach \i in {1,...,4} {
      \foreach \j in {1,...,4} {
        \ifnum\i=2
          \ifnum\j=3
            \def\cellcolor{violet!55}
          \else
            \def\cellcolor{blue!22}
          \fi
        \else
          \ifnum\j=3
            \def\cellcolor{orange!35}
          \else
            \def\cellcolor{white}
          \fi
        \fi
        \pgfmathsetmacro{\x}{\j-1}
        \pgfmathsetmacro{\y}{4-\i}
        \draw[fill=\cellcolor] (\x,\y) rectangle ++(1,1);
        \node at (\x+.5,\y+.5) {$a_{\i\j}$};
      }
    }
    \foreach \i in {1,...,4} {
      \pgfmathsetmacro{\y}{4.5-\i}
      \node[anchor=east] at (-.18,\y) {$P_{\i}$};
    }
    \foreach \j in {1,...,4} {
      \pgfmathsetmacro{\x}{\j-.5}
      \node[anchor=south] at (\x,4.12) {$Q_{\j}$};
    }
    \draw[very thick,blue!65] (0,2) rectangle (4,3);
    \draw[very thick,orange!85!black] (2,0) rectangle (3,4);
    \node[blue!65,anchor=west] at (4.25,2.5)
      {$S^{(2)}_p$, $p\in P_2$};
    \node[orange!85!black,anchor=north] at (2.5,-.18)
      {$S^{(2)}_q$, $q\in Q_3$};
  \end{tikzpicture}
  \caption{The even cross-cover construction for $\lambda=4$. Processes in
  block $P_i\subseteq\mathcal{P}_0$ use row $i$, and processes in
  $Q_j\subseteq\mathcal{P}_1$ use column $j$. The highlighted masking sets
  intersect in $a_{23}$. Each block contains $\lambda/2=2$ processes, so every
  account has two row incidences and two column incidences, for total incidence
  $\lambda=4$.}
  \label{fig:cross-cover-grid}
\end{figure}

\emph{Odd $\lambda$: a cyclic design.} Here $n=(\lambda^2-1)/2$, and an account
cannot split evenly between the classes. The construction alternates accounts
of the two types $(g,h)$ and $(h,g)$, arranged cyclically so that even and odd
differences are covered separately. Set
\[
  g=\frac{\lambda-1}{2},\qquad h=\frac{\lambda+1}{2},\qquad
  n=2gh=\frac{\lambda^2-1}{2},
\]
so that $g+h=\lambda$ and, $g$ and $h$ being consecutive integers,
$\gcd(g,h)=1$. Identify both classes with $\mathbb{Z}_n$, writing
$\mathcal{P}_0=\mathcal{P}_1=\mathbb{Z}_n$, and take all arithmetic modulo $n$.
On a new object instance $O$ place $2n$ accounts $c_t,d_t$ for
$t\in\mathbb{Z}_n$,
specifying each by the processes whose masking set contains it: writing
$X(\cdot)\subseteq\mathcal{P}_0$ and $Y(\cdot)\subseteq\mathcal{P}_1$ for those
two sides,
\[
  X(c_t)=\{t+2hi\}_{i=0}^{g-1},\quad Y(c_t)=\{t+2gj\}_{j=0}^{h-1},
\]
\[
  X(d_t)=\{t+2gj\}_{j=0}^{h-1},\quad Y(d_t)=\{t+1+2hi\}_{i=0}^{g-1}.
\]
Equivalently, $c_t$ occupies the rectangle $X(c_t)\times Y(c_t)$ of the
cross-class grid $\mathcal{P}_0\times\mathcal{P}_1$, and $d_t$ the rectangle
$X(d_t)\times Y(d_t)$. Accordingly $S^{(2)}_x=\{c_t:x\in X(c_t)\}\cup\{d_t:x\in X(d_t)\}$
for $x\in\mathcal{P}_0$, and symmetrically on the other side.

\emph{Sizes~(i) and incidence~(ii).} Fix $x\in\mathcal{P}_0$. Then $x\in X(c_t)$
iff $t=x-2hi$ for some $i\in\{0,\dots,g-1\}$, giving exactly $g$ values of $t$,
and $x\in X(d_t)$ iff $t=x-2gj$ for some $j\in\{0,\dots,h-1\}$, giving $h$; hence
$|S^{(2)}_x|=g+h=\lambda$, and symmetrically $|S^{(2)}_y|=h+g=\lambda$ for
$y\in\mathcal{P}_1$. Each account's incidence is the size of its rectangle's two
sides, $g+h=\lambda$ for $c_t$ and $h+g=\lambda$ for $d_t$.

\emph{Cross-pair cover~(iii).} We show the stronger statement that every pair
$(x,y)\in\mathcal{P}_0\times\mathcal{P}_1$ lies in exactly one account. The
argument for both parities of the difference $y-x$ uses the same observation:
as $\gcd(g,h)=1$, the
map $(i,j)\mapsto gj-hi \bmod gh$ is a bijection $[g]\times[h]\to\mathbb{Z}_{gh}$,
since reducing it modulo $g$ leaves $-hi$ with $h$ invertible, recovering $i$,
symmetrically modulo $h$ recovers $j$, and the Chinese remainder theorem
combines the two.

The differences realized inside $c_t$ are $\{2(gj-hi)\}$ and those inside $d_t$
are $\{1+2(hi-gj)\}$, for $i\in[g]$, $j\in[h]$. Doubling the bijection therefore
makes the former a bijection onto the even residues of $\mathbb{Z}_{2gh}$ and the
latter one onto the odd residues, so $c_t$ realizes each even difference exactly
once and $d_t$ each odd difference exactly once. As $t$ ranges over
$\mathbb{Z}_n$, the unique pair inside $c_t$ realizing a given even difference
shifts cyclically through all $n$ pairs of that difference, and likewise for
$d_t$ on odd differences. Every pair has one parity or the other, so it lies in
exactly one account, and in particular $S^{(2)}_x\cap S^{(2)}_y\ne\emptyset$.
Figure~\ref{fig:cross-cover-cyclic} shows this cover for $\lambda=3$.
\end{proof}

\begin{figure}[ht]
  \centering
  \begin{tikzpicture}[x=1.15cm,y=1.05cm, every node/.style={font=\small}]
    \foreach \x in {0,...,3} {
      \foreach \y in {0,...,3} {
        \pgfmathtruncatemacro{\parity}{mod(\y-\x+4,2)}
        \pgfmathsetmacro{\yy}{3-\x}
        \ifnum\parity=0
          \def\cellcolor{blue!22}
        \else
          \def\cellcolor{orange!35}
        \fi
        \draw[fill=\cellcolor] (\y,\yy) rectangle ++(1,1);
      }
    }
    \node at (.5,3.5) {$c_0$};
    \node at (1.5,3.5) {$d_0$};
    \node at (2.5,3.5) {$c_0$};
    \node at (3.5,3.5) {$d_2$};
    \node at (.5,2.5) {$d_3$};
    \node at (1.5,2.5) {$c_1$};
    \node at (2.5,2.5) {$d_1$};
    \node at (3.5,2.5) {$c_1$};
    \node at (.5,1.5) {$c_2$};
    \node at (1.5,1.5) {$d_0$};
    \node at (2.5,1.5) {$c_2$};
    \node at (3.5,1.5) {$d_2$};
    \node at (.5,.5) {$d_3$};
    \node at (1.5,.5) {$c_3$};
    \node at (2.5,.5) {$d_1$};
    \node at (3.5,.5) {$c_3$};
    \foreach \x in {0,...,3} {
      \pgfmathsetmacro{\yy}{3.5-\x}
      \node[anchor=east] at (-.18,\yy) {$x=\x$};
    }
    \foreach \y in {0,...,3} {
      \pgfmathsetmacro{\xx}{\y+.5}
      \node[anchor=south] at (\xx,4.12) {$y=\y$};
    }
    \node[anchor=east,font=\small] at (-.18,4.38) {$\mathcal P_0$};
    \node[anchor=west,font=\small] at (4.18,4.38) {$\mathcal P_1$};
    \draw[fill=blue!22] (4.55,2.75) rectangle ++(.35,.35);
    \node[anchor=west] at (5.02,2.925) {$c_t$: even $y-x$};
    \draw[fill=orange!35] (4.55,2.15) rectangle ++(.35,.35);
    \node[anchor=west] at (5.02,2.325) {$d_t$: odd $y-x$};
  \end{tikzpicture}
  \caption{The odd cyclic cross-cover for $\lambda=3$, where $g=1$, $h=2$,
  and $n=4$. Rows and columns represent processes $x\in\mathcal P_0$ and
  $y\in\mathcal P_1$. Each cell names the unique account shared by the
  corresponding pair. Accounts $c_t$ cover even cyclic differences and
  accounts $d_t$ cover odd differences. Each process is incident to three
  accounts, and each account occurs in three masking sets.}
  \label{fig:cross-cover-cyclic}
\end{figure}

The construction is available for every $\lambda\ge3$ and, by
Lemma~\ref{lem:cross-cover-compose}, matches the upper bound.

\begin{theorem}[Consensus number under strong untraceability]
\label{thm:multi-cuat-strong-tight}
For every $\lambda\ge 2$,
$\consnum{\cuat_\lambda^+}\le 4\lfloor\lambda/2\rfloor\lceil\lambda/2\rceil$, and
for every $\lambda\ge 3$ equality holds:
\[
  \consnum{\cuat_\lambda^+}= 4\lfloor\lambda/2\rfloor\lceil\lambda/2\rceil
  =\begin{cases}\lambda^2 & \text{if }\lambda\text{ is even,}\\
   \lambda^2-1 & \text{if }\lambda\text{ is odd.}\end{cases}
\]
\end{theorem}

\begin{proof}
The upper bound is Theorem~\ref{thm:multi-phase-strong-refined}. For the lower
bound, let $n=\lfloor\lambda^2/2\rfloor$. Lemma~\ref{lem:cross-cover-exists}
supplies a cross-cover family of order $n$ and
Theorem~\ref{thm:cyclic-one-round} the one-round $n$-consensus protocol each
class runs, so Lemma~\ref{lem:cross-cover-compose} yields a wait-free
$2n$-consensus protocol on $\cuat_\lambda^+$. Finally
$2n=2\lfloor\lambda^2/2\rfloor=4\lfloor\lambda/2\rfloor\lceil\lambda/2\rceil$
matches the upper bound.
\end{proof}

The prime-power condition is therefore confined to attaining the one-round
upper bound through a projective plane, the exact multi-round value carrying no
arithmetic condition on $\lambda$. Moreover, for every $\lambda\ge3$,
$4\lfloor\lambda/2\rfloor\lceil\lambda/2\rceil>\lambda^2-\lambda+1$; hence
multi-round protocols are strictly more powerful than one-round protocols under
strong untraceability.

Having determined $\consnum{\cuat_\lambda^+}$ through the conflict graph
structure, we now examine how that structure affects parallelization
opportunities and fairness guarantees.

\section{Parallelization and fairness}
\label{sec:parallel-fairness}

Sections~\ref{sec:luat} and~\ref{sec:ruat} read the conflict graph at a single
critical configuration, where it bounds how many processes can agree. Taken over
a whole workload instead, the same graph governs how that workload executes: its
chromatic number is the number of sequential rounds a batch of transfers
requires, and its independence number is how many of them commit in one round.
The conflict structure also settles fairness. No $\cuat_\lambda$ implementation
is starvation-free under an unfair scheduler, while $\luat_\lambda$ is, which
locates the asymmetry between the two objects in the scheduler rather than in
the consensus number. The section closes with the attack surface these
quantities expose and the mitigations available against it.

Fix a workload of transfers $\mathcal{T}$ and let $G_{\mathcal T}=(V,E)$ denote
its workload conflict graph: $V$ contains one vertex per transfer (including
distinct transfers that use the same masking set), and two vertices are adjacent
when their masking sets share an index. Adjacent transfers cannot both succeed
concurrently, since the second to commit finds at least one account invalidated
by the first; non-adjacent transfers act on disjoint account indices and
commute. Between sequential rounds, a pending transfer is \emph{rebased} by
replacing every account representation in its masking set with the current
representation of the same index. This preserves the intended logical transfer
while making its masking set valid after earlier rounds.

\begin{theorem}[Chromatic number and sequential complexity]
  \label{thm:chromatic-sequential}
  In an execution without external interference in which every transfer remains
  balance-eligible when rebased, the chromatic number
  $\chi(G_{\mathcal T})$ equals the minimum number of sequential rounds required
  to execute every transfer in $\mathcal{T}$ successfully, when pending
  transfers may be rebased between rounds.
\end{theorem}

\begin{proof}
  ($\le$) A proper $k$-coloring of $G_{\mathcal T}$ partitions $V$ into $k$
  independent sets $V_1,\dots,V_k$. Before round $i$, rebase the transfers in
  $V_i$ onto the current account representations and schedule them concurrently.
  No two transfers in $V_i$ share an index, so their $\tau$-invocations act on
  pairwise disjoint accounts and all succeed.

  ($\ge$) Conversely, take any successful schedule using $k$ rounds. Assign to
  each transfer the index of its round. Two transfers in the same round cannot
  share an index, since rebasing preserves indices and one of the two would
  otherwise fail. The assignment is therefore a proper $k$-coloring, so
  $\chi(G_{\mathcal T})\le k$. Hence no schedule uses fewer than
  $\chi(G_{\mathcal T})$ rounds.
\end{proof}

The chromatic number quantifies sequential rounds; the dual quantity, the maximum independent set, quantifies parallel throughput.

\begin{theorem}[Maximum independent set and parallel throughput]
  \label{thm:mis-parallel}
  The independence number $\alpha(G_{\mathcal T})$ equals the maximum number of
  transfers from $\mathcal{T}$ that can execute in parallel in a single round
  with all succeeding.
\end{theorem}

\begin{proof}
  ($\ge$) Any independent set $I\subseteq V$ is a set of pairwise
  non-conflicting transfers, so the masking sets of $I$ are pairwise disjoint in
  the index sense; their $\tau$-invocations commute and all succeed when
  scheduled concurrently. Hence $\alpha(G_{\mathcal T})$ transfers fit in one
  round.

  ($\le$) Conversely, any set of transfers that all succeed concurrently must
  be pairwise non-conflicting (the preceding theorem's argument applied to one
  round) and so forms an independent set in $G_{\mathcal T}$. The maximum such
  set has size $\alpha(G_{\mathcal T})$.
\end{proof}

The conflict graph also exposes a fundamental limitation: when
processes use overlapping masking sets, at most one transfer
succeeds, forcing coordination. While CUAT provides
\emph{wait-freedom}, this only guarantees termination. Transfers may
fail due to concurrent interference on overlapping masking sets, even
when the invoking process has sufficient balance. This structural
property enables starvation under adversarial scheduling.

We consider the standard asynchronous model with an adversarial scheduler
that may delay or advance any correct process arbitrarily, subject only to
preserving program order within each process. Under such a scheduler, we
show that CUAT cannot guarantee even minimal fairness. Specifically, no
implementation can prevent starvation: an adversarial scheduler may
indefinitely deny successful transfers to some processes while allowing
others to succeed repeatedly. While denial-of-service attacks on
randomization-based systems, specifically \emph{Quisquis}, have been conjectured
by B{\"{u}}nz \etal~\cite{DBLP:conf/fc/BunzAZB20}, we prove this
impossibility formally.

\begin{definition}[Starvation-freedom]
  \label{def:lockout}
  An implementation of an untraceable asset-transfer object is
  \emph{starvation-free} if, for every correct process $p$, every infinite
  execution in which $p$ invokes infinitely many eligible transfers contains
  infinitely many successful invocations by $p$. Eligibility is evaluated at
  invocation: the sender account is owned by $p$, its balance covers the chosen
  value, and the masking set is valid. For $\cuat_\lambda$, successive
  invocations may follow the current representations of the same logical
  accounts through rerandomization; for $\luat_\lambda$, they may spend
  successive unspent accounts owned by $p$.
\end{definition}
This notion is deliberately minimal: it requires neither bounded
waiting nor proportional success rates, only that a process which
persistently attempts transfers and is frequently eligible cannot be
starved indefinitely.
\begin{theorem}[Impossibility of starvation-freedom]
  \label{thm:no-fairness}
  No $\cuat_\lambda$ implementation can guarantee starvation-freedom
  under an asynchronous, unfair
  scheduler. Specifically, there exists a schedule in which one correct
  process $p^\star$ succeeds infinitely often along the rerandomized lineage of
  a masking set $S$, while another correct contender on that lineage is
  starved: its infinitely many
  invocations all return $\false$. This holds even when balances pose no
  fundamental limit, since zero-value transfers ($v=0$) keep $p^\star$
  perpetually eligible.
\end{theorem}

\begin{proof}
  Fix a masking set $S=\{a_1,\dots,a_\lambda\}$ with two distinct accounts $a,a^\star\in S$ owned by processes $p$ and $p^\star$, so $\mu(a)=\{p\}$ and $\mu(a^\star)=\{p^\star\}$. Let $b,b^\star\in S$ be arbitrary destinations. We construct an infinite execution in \emph{rounds} $r=1,2,\dots$ where both processes repeatedly attempt transfers on the current representations of these logical accounts, and the scheduler resolves the race in favor of $p^\star$ every round. To avoid cumbersome indices, within each round we write $S,a,a^\star,b,b^\star$ for those current representations.

  \emph{Round $r$.} Process $p$ invokes
  $\tau_r=\tr{S'_p,\pi_p}$ with $\pi_p=(S,a,b,v,p)$ for some
  $v\le\beta(a)$. At invocation, $S$ is valid and the transfer is eligible. The
  scheduler pauses $p$ before its linearization point and lets $p^\star$
  complete its zero-value transfer
  $\tau^\star_r=\tr{S'_\star,\pi^\star}$ with witness
  $\pi^\star=(S,a^\star,b^\star,0,p^\star)$. This invocation succeeds:
  $S\subseteq\mathcal{A}$ holds at its linearization point, and
  $\beta(a^\star)\ge 0$ trivially. It invalidates all of $S$ and publishes the
  newly generated masking set $S'_\star$, so the post-state has
  $S\cap\mathcal{A}=\emptyset$.

  The scheduler now resumes $p$. Because $S\cap\mathcal{A}=\emptyset$ at
  $\tau_r$'s linearization point, the precondition
  $S\subseteq\mathcal{A}$ fails and $\tau_r$ returns $\textsc{false}$.

  At the start of round $r+1$, $p$ rebases its next invocation onto the latest
  valid masking set containing the current descendant of $a$, and the pattern
  repeats. The scheduler iterates the same protocol indefinitely.

  \emph{Why $p^\star$ stays eligible.} $p^\star$'s value parameter is $v^\star=0$, so its balance is preserved across rounds. Its accounts after each round are replaced by the rerandomization of $S$. Hence $p^\star$ never runs out of eligible transfers, and the schedule is well-defined for all $r\ge 1$.

  \emph{Conclusion.} $p^\star$ succeeds infinitely often (one success per round) while every invocation of $p$ returns $\textsc{false}$. Starvation arises not from balance exhaustion or crash failure, but from the unfair scheduling of an adversary who can always insert $p^\star$'s commit before $p$'s on the shared masking set.
\end{proof}

The impossibility is specific to the constant-state object. Its linear-state
counterpart is starvation-free, for the same reason that its consensus number
does not depend on $\lambda$ (Corollary~\ref{cor:luat-untraceability-free}): a
$\luat$ transfer consumes one account rather than its whole masking set, so a
competitor cannot invalidate it.

\begin{theorem}[$\luat$ guarantees starvation-freedom]
  \label{thm:luat-fairness}
  Every $\luat_\lambda$ implementation is starvation-free. If a correct process
  repeatedly invokes transfers spending distinct unspent accounts it owns, with
  masking sets valid at invocation, then infinitely many of those invocations
  succeed, under any asynchronous scheduler.
\end{theorem}

\begin{proof}
  Let $p$ repeatedly invoke transfers, the $r$-th spending an account $a_r$ it
  owns whose nullifier $\nu_r=\mathsf{Nullify}(a_r)$ has never entered the
  deny-set. We show each such invocation succeeds. Its preconditions at the
  linearization point are: $S_r\subseteq\mathcal{A}_{\text{allow}}$, which holds
  because $S_r$ was valid at invocation and the allow-set only grows;
  $p\in\mu(a_r)$, which is local to $p$; $\nu_r\notin\mathcal{N}$, which holds
  because only a transfer spending $a_r$ can add $\nu_r$, distinct accounts have
  distinct nullifiers by collision-resistance of $\mathsf{Nullify}$, and no other
  process owns $a_r$; $\mathcal{A}_{\text{new},r}$ contains only previously
  unused representations, by account-representation uniqueness; and the
  balance condition, which holds because a transfer never
  alters balances on the existing allow-set. Every precondition is therefore
  insensitive to concurrent activity, so the invocation succeeds regardless of the
  schedule, and all infinitely many of them do.
\end{proof}

The two objects differ in what a scheduler can revoke. Under $\luat$ the deny-set
is monotone and a nullifier is consumed only by an owner of its account, so a
process that is eligible remains eligible; under $\cuat_\lambda$ a competitor's
commit invalidates the entire masking set, and eligibility is exactly what the
scheduler controls. The state growth of Section~\ref{subsec:luat-gc} is what
buys this, since spent accounts are never reclaimed and a pending transfer's
masking set cannot be invalidated. Constant state and starvation-freedom are thus
in tension, and $\lambda$ does not mediate it in either direction. Escaping the
$\cuat_\lambda$ conflict is moreover not sufficient for asynchronous progress.
Anonymous Zether~\cite{DBLP:conf/fc/BunzAZB20} makes invalidation sender-local
and so admits no competitor-driven starvation, yet bounds its deny-set by
expiring nullifiers at epoch boundaries, which lets a scheduler that defers
inclusion past every boundary starve a process just as effectively. Theorem~\ref{thm:luat-fairness} rests on eligibility being irrevocable once
held, and a deadline revokes it no less than a competitor does.

The three theorems above support three correlated attack vectors against an honest batch.
\begin{itemize}
  \item \emph{Throughput attack.} By submitting transfers whose masking sets
    intersect honest ones, an adversary can increase $\chi(G_{\mathcal T})$ and
    hence the number of sequential rounds required to commit the enlarged
    workload. Although adding vertices cannot decrease
    $\alpha(G_{\mathcal T})$, adversarial transfers that a scheduler admits in
    a round can exclude adjacent honest transfers from that round and thereby
    reduce honest throughput.
  \item \emph{Starvation attack.} On any chosen masking set, an adversary in concert with an unfair scheduler can starve the honest contenders by issuing repeated zero-value transfers. The attacker's eligibility is renewed each round without cost.
  \item \emph{Amplification through masking-set incidence.} Increasing
    $\lambda$ gives each adversarial transfer more account indices through which
    it may conflict with honest transfers, but its actual reach is determined by
    the incidence of those indices in the workload. Strong untraceability
    regularizes this incidence rather than eliminating it, so privacy and
    contention remain coupled through the conflict graph.
\end{itemize}

B\"{u}nz \etal~\cite{DBLP:conf/fc/BunzAZB20} conjectured this attack pattern against \emph{Quisquis}; the three theorems characterize it precisely.

\paragraph{Instantiations.}
In a deployment, a masking set is an anonymity set or ring, invalidation is the
publication of a nullifier or key image, and a conflict is the rejection of a
transfer whose inputs another transfer has already consumed.
Sui~\cite{suilutris2024} and Zef~\cite{zef2023} commit transfers on disjoint
accounts in parallel and invoke the BFT layer only on contention; $\chi(G)$ and
$\alpha(G)$ bound what that architecture yields once the transferred object is
untraceable rather than public, and Theorem~\ref{thm:no-fairness} denies its
fast path starvation-freedom. Anonymous Zether~\cite{DBLP:conf/fc/BunzAZB20} is a
documented case of a design that met this conflict and avoided it. Its authors
first consider incrementing a nonce for every account of the anonymity set, and
observe that a transfer is then rejected whenever another transfer touches one
of its accounts first, which is exactly an edge of $G$. They reject that scheme
and derive the nonce from the sender's secret key instead, so a transfer
consumes only its own sender's nonce and leaves the rest of the set spendable.
Invalidation becomes sender-local, transfers commute, and the object is
$\luat_\lambda$ rather than $\cuat_\lambda$. The nonce set is a deny-set,
discarded at each epoch boundary, so the state stays bounded without account
replacement. The price is a concurrency bound: the nonce set admits one nonce
per account per epoch, so each account has exactly one transfer in flight at a
time and the epoch length fixes its throughput. This is the per-account form of
the rate limiting listed below, adopted here to bound the deny-set rather than
to resolve contention. The contrast with $\luat_\lambda$
is visible in the same setting. A Monero ring or a Zcash note-commitment set is
referenced but not consumed, so two transfers sharing decoys both commit and the
conflict graph has no edges; a repeated key image or nullifier is a double spend
rather than contention. This is the deployed form of the separation of
Section~\ref{subsubsec:luat-untr-cons}: the masking set carries privacy in both
objects, but only in $\cuat_\lambda$ does it also carry synchronization.

The proof of Theorem~\ref{thm:no-fairness} uses three ingredients: a wait-free
implementation, a scheduler that decides the order of steps, and transfers that a
process may repeat at no cost. A deployment that wants starvation-freedom must
weaken one of them. The main options are the following.
\begin{itemize}
  \item \emph{Synchrony assumptions.} Under a scheduler with bounded step delays,
    or under any fair scheduler, an adversary can no longer take unboundedly many
    steps between two steps of its victim, which is what the proof needs; the
    progress guarantee then holds only as long as the timing assumption does.
  \item \emph{Per-transfer fees.} Charging for each attempt, as Bitcoin and
    Ethereum do, limits how many transfers an adversary can issue, but honest
    users pay the fee as well.
  \item \emph{Scheduler with first-contender tracking.} Recording which transfer
    reached a masking set first, by timestamp or sequence number, breaks the
    symmetric race the proof relies on and keeps wait-freedom, but the recorded
    order reveals timing.
  \item \emph{Per-set rate limiting.} Allowing a masking set to be attempted only
    a fixed number of times per time window stops an adversary from repeating
    attempts on it, but restricts honest users in the same way.
  \item \emph{Obstruction- or lock-freedom with randomized backoff.} Replacing
    wait-freedom by a weaker progress condition~\cite{DBLP:conf/icdcs/HerlihyLM03}
    gives starvation-freedom with high probability, but no longer guarantees that
    every operation terminates.
\end{itemize}
Each option drops a different requirement, and none of them keeps all three
ingredients.

\paragraph{Practical implications.}
The quantities $\chi(G_{\mathcal T})$ and $\alpha(G_{\mathcal T})$ depend on the
workload, not on the object alone, and a deployment can act on them. At a fixed
$\lambda$, it chooses which transfers are batched together and in what order, and
so changes both quantities. It does not choose the range they can take: strong
untraceability forces uniform incidence, and uniform incidence keeps the conflict
graph dense. Fairness cannot be tuned at all. By Theorem~\ref{thm:no-fairness} a
deployment has to accept one of the options above, and each of them gives up
asynchrony, wait-freedom, a fee-free interface, timing privacy, or the equal
treatment of honest users. A system that wants to keep all five should not use its anonymity
mechanism to bound the state: $\luat_\lambda$ has consensus number $2$ for every
$\lambda$, is starvation-free, and by
Theorem~\ref{thm:part-consensus-number} its storage can be reclaimed without
additional coordination. Both families pay for sender privacy; they differ in
whether they pay in storage or in synchronization.

\section{Conclusion}\label{sec:conclusion}
We have located both families of untraceable cryptocurrency in the consensus
hierarchy. The comparison amounts to a trade-off between state and
synchronization.

The linear-state object $\luat_\lambda$, which models Zcash and Monero, has
consensus number $2$ independently of $\lambda$, and $\max(k,2)$ when $k$
processes share an account. Its transfers on distinct accounts commute, but
their responses expose enough of the resulting deny-set to distinguish their
positions in the linearization order and solve two-process consensus. This
information, absent from standard asset transfer, accounts for the floor at
$2$. The object is starvation-free, and the partitioning object that
garbage-collects its state is equivalent to fetch-and-add and hence also of
consensus number $2$, so garbage collection carries no synchronization cost.
Drawing the partition from a randomness beacon, as resistance to grinding
requires, leaves that number unchanged and trades the ordering rule for an
assumption on the number of correct processes. The cost of $\luat_\lambda$ is
storage.

The constant-state object $\cuat_\lambda$, which models \emph{Quisquis} and
similar systems, removes that cost. Under weak untraceability it has unbounded
consensus number already with one-round protocols, via a hub-account family.
Under strong untraceability we prove an equivalence: the object satisfies strong
untraceability on a history exactly when any two accounts sharing a masking set
appear in the same number of the masking sets occurring in that history. This
uniformity caps the one-round upper bound at
$\lambda^2-\lambda+1$, attained by the line set of $\mathrm{PG}(2,\lambda-1)$
when $\lambda-1$ is a prime power, and caps $\consnum{\cuat_\lambda^+}$ at
$4\lfloor\lambda/2\rfloor\lceil\lambda/2\rceil$. The latter bound we attain for
every $\lambda\ge 3$, composing a per-class one-round protocol built from a
cyclic difference cover with a grid or cyclic cross-cover family, so
$\consnum{\cuat_\lambda^+}=\lambda^2$ for even $\lambda$ and $\lambda^2-1$ for
odd, with no arithmetic condition on $\lambda$. The same cyclic construction
gives the only one-round lower bound available when $\lambda-1$ is not a prime
power, placing the one-round regime between $\lfloor\lambda^2/2\rfloor$ and
$\lambda^2-\lambda+1$. Finally, $\cuat_\lambda$ cannot guarantee
starvation-freedom under asynchronous scheduling: an unfair scheduler starves any
process via repeated zero-value transfers on its masking set.

Both objects satisfy the same uniformity characterization, which locates the
difference between them. Untraceability constrains which masking sets either
object may use; $\consnum{\luat_\lambda}$ does not depend on that choice, whereas
$\consnum{\cuat_\lambda^+}$ is determined by it, because a $\luat$ transfer names
its masking set as decoys while a $\cuat$ transfer consumes it. Privacy is thus
not by itself costly in synchronization. It becomes costly when the anonymity
mechanism is also made to carry the state bound, at a price quadratic in the
privacy parameter.

Several directions remain open.
\begin{itemize}
  \item \emph{Other privacy mechanisms.} We analyze masking-set untraceability,
    in its linear-state and constant-state forms. Mix nets, ORAM-style designs
    and related mechanisms fall outside both objects, and their synchronization
    cost has to be determined separately.
  \item \emph{Designs between the two objects.} $\luat_\lambda$ has linear state
    and consensus number $2$, and $\cuat_\lambda^+$ has constant state and
    consensus number $4\lfloor\lambda/2\rfloor\lceil\lambda/2\rceil$. Designs
    whose state grows sublinearly in the number of transfers, through batched
    revocation or hierarchical commitments, lie between the two and are not yet
    placed in the hierarchy. Anonymous Zether~\cite{DBLP:conf/fc/BunzAZB20}
    occupies part of the gap: it keeps the $\luat_\lambda$ semantics and
    consensus number $2$, and bounds its deny-set by one epoch of activity. It
    pays in concurrency rather than in state, since an epoch-scoped nullifier
    admits one nonce per account per epoch, so each user has one transfer in
    flight and the epoch length fixes per-account throughput. Whether the gap
    can be closed without such a cap, and without the synchrony that epochs
    require, is open.
  \item \emph{Fairness under weaker assumptions.} Theorem~\ref{thm:no-fairness}
    assumes wait-freedom, an asynchronous scheduler, and transfers that cost
    nothing to repeat, and Section~\ref{sec:parallel-fairness} lists the ways to
    give up one of them. Which relaxation buys starvation-freedom at the smallest
    price is open, and in particular whether obstruction- or
    lock-freedom~\cite{DBLP:conf/icdcs/HerlihyLM03} with randomized backoff gives
    starvation-freedom with high probability without weakening untraceability.
  \item \emph{Authenticated data structures.} Both objects are maintained on
    structures whose concurrent behavior is itself unsettled: Merkle trees for
    the allow-set, and accumulators for membership proofs. Their updates and
    proof computations are specified sequentially, and how much of that
    maintenance can proceed without agreement is not known. Christ and
    Bonneau~\cite{DBLP:conf/fc/ChristB23} bound the storage and proof-update cost
    of such structures; the corresponding question for synchronization is open.
\end{itemize}

Our results measure privacy the way the consensus hierarchy measures any other
requirement, by where the object that provides it lands. On that scale, sender
untraceability is free only as long as it is not also made to bound the state;
where it is, the consensus number is quadratic in $\lambda$.

\printbibliography

 \end{document}